\documentclass{article}
\usepackage{graphicx}
\usepackage{caption}
\usepackage{amssymb}
\usepackage{amsmath}
\usepackage{subfigure}
\usepackage{authblk}
\usepackage{amsthm}

\newcommand{\deleted}[1]{}
\newcommand{\kleq}{({{\leq}k})}
\newcommand{\rleq}{({{\leq}r})}
\newtheorem{definition}{Definition}
\newtheorem{theorem}{Theorem}[section]
\newtheorem{lemma}[theorem]{Lemma}
\newtheorem{corollary}[theorem]{Corollary}
\newtheorem{claim}[theorem]{Claim}

\author[1]{Evanthia Papadopoulou}
\author[1]{Maksym Zavershynskyi}
\affil[1]{Faculty of Informatics, Universit\`{a} della Svizzera italiana, Lugano, Switzerland \{evanthia.papadopoulou,maksym.zavershynskyi\}@usi.ch}

\begin{document}
\title{The Higher-Order Voronoi Diagram of\\ Line Segments
}

\date{}

\maketitle

\begin{abstract}
  \let\thefootnote\relax\footnotetext{A preliminary version appeared in \emph{Proc.  $23^{rd}$ International Symposium of Algorithms and Computation} (ISAAC) 2012, LNCS 7676: 177-186.\\ Supported in part by the Swiss National Science Foundation grant 200021-127137  and the ESF EUROCORES program EuroGIGA/VORONOI, SNF 20GG21-134355.}
 Surprisingly, the order-$k$ Voronoi diagram of line segments had received no attention in the computational-geometry literature. 
It illustrates properties surprisingly different from its counterpart for points; 
for example, a single order-$k$ Voronoi region may consist of $\Omega(n)$ disjoint faces.
We analyze the structural properties of this diagram and show that its combinatorial complexity for $n$ non-crossing line segments is $O(k(n-k))$, despite the disconnected regions.
  The same bound holds for $n$ intersecting line segments, when $k\geq n/2$.
  We also consider the order-$k$ Voronoi diagram of line segments that form a planar straight-line graph, and augment the definition of an order-$k$ Voronoi diagram to cover non-disjoint sites, addressing the issue of non-uniqueness for $k$-nearest sites. 
 Furthermore, we enhance the iterative approach to construct this diagram.
  All bounds are valid in the general $L_p$ metric, $1\leq p\leq \infty$. 
  For non-crossing segments in the $L_\infty$ and $L_1$ metrics, we show a tighter $O((n-k)^2)$ bound for $k>n/2$.

\end{abstract}

\section{Introduction}
Given a set of $n$ simple geometric objects in the plane, called sites, the \emph{order-$k$ Voronoi diagram} of $S$ is a partitioning of the plane into regions, such that every point within a given order-$k$ region has the same $k$ nearest sites.  
For $k=1$, we derive the classic \emph{nearest-neighbor Voronoi diagram}, and for $k=n-1$, the \emph{farthest-site Voronoi diagram}.  

For point-sites, the order-$k$ Voronoi diagram has been studied extensively in the computational geometry literature, see e.g.,~\cite{Agarwal98,Aurenhammer92,Chan00,Chazelle87,Edelsbrunner87,Lee82,Liu11,Ramos99} and~\cite{Sack00} for a survey.  
Surprisingly,  it has been largely ignored for any other type of site.
This is the case even for simple line segments, which nevertheless, play a fundamental role in applications involving polygonal objects in the plane.  
See e.g.,~\cite{Papadopoulou11} for applications of higher-order line-segment Voronoi diagrams in deriving the \emph{Probability of Fault} in a VLSI design under \emph{random manufacturing defects}. 
Even the farthest line-segment Voronoi diagram ($k=n-1$) has only recently been considered by Aurenhammer et al.~\cite{Aurenhammer06}, showing properties surprisingly different from its counterpart for points.
Only a few additional types of farthest Voronoi diagrams for generalized sites have been considered in the literature, see e.g.,~\cite{Cheong11,Papadopoulou12,Rappaport89}, and the \emph{farthest abstract Voronoi diagram}~\cite{Mehlhorn01}. 
To the best of our knowledge, no prior work exists on order-$k$ Voronoi diagrams, $1<k<n-1$, for generalized sites, other than points and additively weighted points~\cite{Rosenberger91}.

In this paper we investigate combinatorial properties and the basic iterative construction for the order-$k$ Voronoi diagram of line segments.
We first establish  complexity results for disjoint line segments, and then extend our investigation to line segments forming a \emph{planar straight-line graph} (PSLG,  in short) and arbitrary line segments that may  intersect.
Although a single order-$k$ Voronoi region may disconnect into $\Omega(n)$ disjoint faces,  we show that the combinatorial complexity of the diagram for $n$ non-crossing line segments remains $O(k(n-k))$, as in the case of points where no disconnected regions exist.
In addition, the union of all faces affiliated with a segment $s$ forms a connected region, which is \emph{weakly star-shaped} with respect to $s$.
The case of a PSLG is particularly interesting. 
On one hand, this is important for applications involving polygonal objects in the plane or embedded planar graphs (see e.g.,~\cite{Papadopoulou11}), and on the other, it introduces new requirements for the definition of an order-$k$ Voronoi diagram.
Segments in a planar straight-line graph are not disjoint, thus, the standard definition of an order-$k$ Voronoi diagram for disjoint sites is not sufficient, because areas may exist, which are equidistant from multiple sites whose number is independent of $k$.
We augment the standard definition of an order-$k$ Voronoi diagram to include sites that are not disjoint, such as segments forming a planar straight-line graph, resolving the issue of non-uniqueness for $k$-nearest sites. 
For arbitrary line segments that may intersect, we show that intersections only affect the asymptotic complexity of the diagram for small $k, k<n/2$.
For $k\geq n/2$, the asymptotic complexity of the diagram is independent of the number of intersections, i.e., it remains $O(k(n-k))$.
We also extend our results to the general $L_p$ metric, $1\leq p\leq \infty$, and show  a tighter bound of $O((n-k)^2)$ for $k>n/2$ and non-intersecting line segments.
To construct the diagram we revisit the standard iterative construction within the standard time complexity and discuss some interesting problems due to the presence of disconnected regions.
A plane sweep approach is given in~\cite{Zavershynskyi13}.

In a subsequent companion paper~\cite{Bohler13}, we generalize the combinatorial results for disjoint line segments to  \emph{higher-order abstract Voronoi diagrams} and refine their complexity bound to $2k(n-k)$.
Non-disjoint line segments, such as line segments forming a PSLG  and intersecting line segments addressed in this paper, do not fall under the umbrella of  abstract Voronoi diagrams in our companion paper, as their bisectors do not comply with the  axioms of the underlying system of bisectors. 

This paper is organized as follows.
Preliminaries and definitions are given in Section~\ref{section:preliminaries}. 
In Section~\ref{section:disconnected}, we show the presence of disconnected regions, where a single region can disconnect to $\Omega(n)$ faces. 
In Section~\ref{section:structural_complexity}, we establish the structural complexity of the order-$k$ Voronoi diagram of disjoint line segments.
These combinatorial results are extended to intersecting line segments in Section~\ref{section:intersecting_segments}, and to the general $L_p$ metric, $1\leq p\leq \infty$,  in Section~\ref{section:lp}.
In Section~\ref{section:abutting_segments}, we consider the order-$k$ Voronoi diagram of line segments forming a planar straight-line graph and augment the definition of an order-$k$ Voronoi diagram to cope with non-disjoint sites. 
In Section~\ref{section:algorithms}, we enhance the iterative construction to construct the order-$k$ line segment Voronoi diagram and in Section~\ref{section:conclusion} we conclude.

\section{Preliminaries}\label{section:preliminaries}

Let $S=\{s_1,s_2,\ldots,s_n\}$ be a set of $n$ line segments in $\mathbb{R}^2$. 
Line segments are assumed disjoint in Sections~\ref{section:preliminaries}--\ref{section:structural_complexity}, but they may touch at endpoints or intersect in subsequet sections. 
Unless stated otherwise, we make the \emph{general position assumption} (applicable to disjoint line segments) that no more than three sites touch the same circle and no more than two endpoints lie on the same line.

The distance between a point $p$ and a line segment $s$ is measured as the minimum Euclidean distance between $p$ and any point on $s$, $d(p,s)=\min_{q\in s}d(p,q)$, where $d(p,q)$ is the Euclidean distance between two points $p$, $q$.
The bisector of two segments $s_i$ and $s_j$ is the locus of points equidistant from both segments, i.e., $b(s_i,s_j)=\{x~|~d(x,s_i)=d(x,s_j)\}$. 
The bisector of two disjoint line segments is a curve which consists of a constant number of line segments, rays, and parabolic arcs.  

Let $H \subset S$. 
The generalized Voronoi region of $H$, $V(H,S)$, is the locus of points that are closer to all segments in $H$ than to any segment not in $H$:
\begin{equation}
  V(H,S)=\{x~|~\forall s\in H, \forall t\in S\setminus H~d(x,s)<d(x,t)\} \label{eq:region}
\end{equation}
For $|H|=k$, $V(H,S)$ is the order-$k$ Voronoi region of $H$, denoted $V_k(H,S)$.
\begin{equation}
  V_k(H,S)=V(H,S)\mbox{~for $|H|=k$}  \label{eq:kregion}
\end{equation}

The partitioning of the plane into order-$k$ Voronoi regions gives the order-$k$ Voronoi diagram of $S$, $\mathcal{V}_k(S)$.
Note that an order-$k$ Voronoi region is only defined for $|H|=k$.
A maximal interior-connected subset of a region is called a face.
For $k=n-1$ we have the farthest Voronoi diagram of $S$, denoted as $\mathcal{V}_f(S)$.
A farthest Voronoi region of a segment $s \in S$ is $V_f(s,S)=V_{n-1}(S\setminus\{s\},S)$.
Fig.~\ref{fig:2-order_segments_disc} illustrates an example of an order-2 Voronoi diagram of line segments.

The following lemma is a simple generalization of~\cite{Aurenhammer06} for $1\leq k\leq n-1$.

\begin{lemma}\label{lemma:supp}
  Consider a face $F$ of region $V_k(H,S)$.
  $F$ is unbounded (in the direction $r$) iff there exists an open halfplane (normal to $r$) that intersects all segments in $H$ but no segment in $S\setminus H$. 
\end{lemma}
\begin{proof}
  ($\Rightarrow$)
  Let $F$ be an unbounded face of region $V_k(H,S)$.
  Let $x\in V_k(H,S)$, and let $r$ be a ray emanating from $x$ to an unbounded direction of the face.
  Since $x\in V_k(H,S)$, $x$ is the center of the open disk that intersects all segments in $H$ and does not intersect segments in $S\setminus H$.
  While we move $x$ along $r$ towards infinity, the disk expands until it becomes an open halfplane that intersects all segments in $H$ but no segment in $S\setminus H$.
  Thus, such a halfplane exists.

  ($\Leftarrow$)
  Let $h$ be an open halfplane that intersects all segments in $H$ but no segment in $S\setminus H$.
  Let $h'$ be the open halfplane $h$ translated parallel to itself until one of the segments $s\in H$ stops intersecting $h$.
  At this moment, $s$ touches the boundary of $h'$ at some point $x$.
  Consider the ray $r$ in $h$ emanating from $x$ orthogonal to the boundary of $h$.
  Let $D$ be a disk centered at an arbitrary point $y$ on $r$, which intersects all segments in $H$.
  Then, $D\subset h$, which means that $D$ does not intersect any segment in $S\setminus H$.
  Therefore, $y\in V_k(H,S)$.
  Since the point $y\in r$ was taken arbitrarily, the ray $r$ is entirely enclosed in $V_k(H,S)$, i.e., the Voronoi region is unbounded in this direction. 
\end{proof}

\begin{definition}
  \label{def:supporting-halfplane}
  A \emph{supporting halfplane} of segments $s_1,s_2 \in S$ and $H \subseteq S$, where $s_1, s_2 \not \in H$, is an open halfplane $h$ whose boundary passes through endpoints of $s_1,s_2$ (at least one endpoint of each segment), with the property that $h$ intersects all segments in $H$ but no segment in $S \setminus H$.
\end{definition}

\begin{corollary} \label{cor:unb_edge}
  (of Lemma~\ref{lemma:supp}) There is an unbounded Voronoi edge separating regions $V_k(H\cup\{s_1\},S)$ and $V_k(H\cup\{s_2\},S)$ if and only if there is a halfplane supporting $s_1$, $s_2$, and $H$.
\end{corollary} 

\section{Disconnected Regions}\label{section:disconnected}
For line segments, a single order-$k$ Voronoi region  may be disconnected and it may consist of multiple disjoint faces, unlike its counterpart for  points.
For example  in Fig.~\ref{fig:2-order_segments_disc}, the  order-2 Voronoi region of the pair of line segments shown in bold, consists of two faces, which are shown shaded. 
This phenomenon was first  pointed out by Aurenhammer et al~\cite{Aurenhammer06} for the farthest line segment Voronoi diagram, where a single Voronoi region was shown possible to disconnect into $\Theta(n)$ faces in the worst case.
\begin{figure}
  \begin{center}
    \includegraphics[width=2.5in, clip, trim = 0mm 10mm 0mm 22mm]{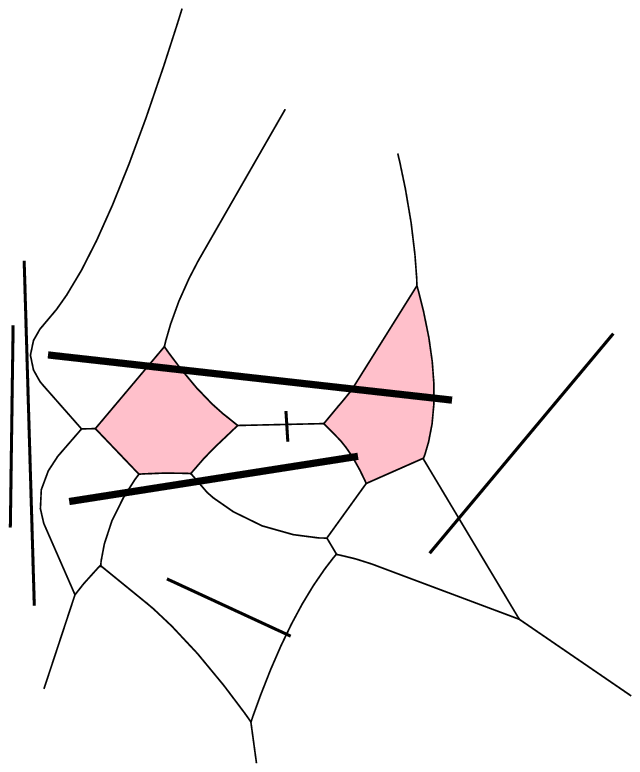}
  \end{center}
  \vspace{-6mm}
  \caption{Order-2 Voronoi diagram with a region that breaks into two disconnected faces, induced by the same pair of sites.}
  \label{fig:2-order_segments_disc}
\end{figure}
\begin{lemma}
  For $k>1$, an order-$k$ region of $\mathcal{V}_k(S)$ can have $\Omega(n)$ disconnected faces in the worst case.
\end{lemma}

\begin{proof}
  We first describe an example where an order-$k$ Voronoi region is disconnected into $n{-}k{-}1$ bounded and two unbounded faces. 
  Consider a set $H$ of $k$ almost parallel long segments.
  These segments induce a region $V_k(H,S)$. 
  Consider the minimum disk that intersects all segments in $H$, and moves along their length. 
  We place the remaining $n{-}k$ segments of $S\setminus H$ in such a way that they create obstacles for the disk. 
  While the disk moves along the tree of $\mathcal{V}_f(H)$, it intersects the segments of $S\setminus H$ one by one and creates $\Omega(n-k)$ disconnectivities (see Fig.~\ref{fig:disc_segm01}).  
  In particular, $V_k(H,S)$ has $n{-}k{-}1$ bounded and two unbounded faces.
  \begin{figure}
    \begin{center}
      \includegraphics[width=4 in, clip, trim = 0mm 0mm 0mm 0mm]{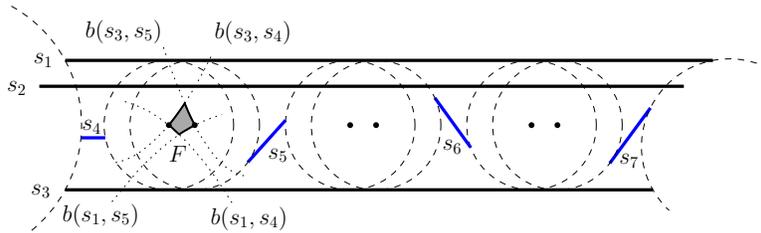}
    \end{center}
    \caption{The obstacles in between the long segments $H$ induce $n-k+1$ disconnectivities in the region of $V_3(H,S)$, $H=\{s_1,s_2,s_3\}$, $k=3$.
  The face $F\subset V_3(H,S)$ is enclosed in between bisectors $b(s_3,s_5)$, $b(s_3,s_4)$, $b(s_1,s_5)$ and $b(s_1,s_4)$.}
  \label{fig:disc_segm01}
\end{figure}

We now follow~\cite{Aurenhammer06} and describe an example in which an order-$k$ Voronoi region is disconnected into $k$ unbounded faces. 
Consider $n-k$ segments in $S\setminus H$, degenerated into points and placed close to each other.
The remaining  $k$ non-degenerate segments in $H$ are organized in a cyclic fashion around them (see Fig.~\ref{fig:disc_segm02}). 
Consider a directed line $g$ through one of the degenerate segments $s$. 
Rotate $g$ around $s$ and consider the open halfplane to the left of $g$. 
During the rotation, the positions of $g$, in which the halfplane intersects all  $k$ segments, alternate with the positions in which it does not (see Fig.~\ref{fig:disc_segm02}(a)).
The positions at which the halfplane touches endpoints of non degenerate segments correspond to unbounded Voronoi edges, such as $g(s_5,s_3)$ and $g(s_5,s_4)$ in Fig.~\ref{fig:disc_segm02}(b), that define an unbounded Voronoi face of $V_k(H,S)$.
Each pair of consecutive unbounded Voronoi edges bounds a distinct unbounded face. 
Each unbounded edge corresponds to a halfplane that touches an endpoint of a line segment in $H$.
Thus, the number of unbounded faces of $V_k(H,S)$ is $|H|=k$. 

\begin{figure}
  \begin{center}
    \subfigure[]{\includegraphics[width=2.3 in]{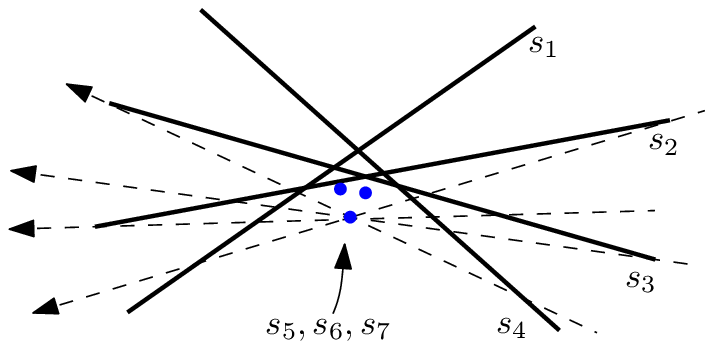}}
    \subfigure[]{\includegraphics[width=3.5 in]{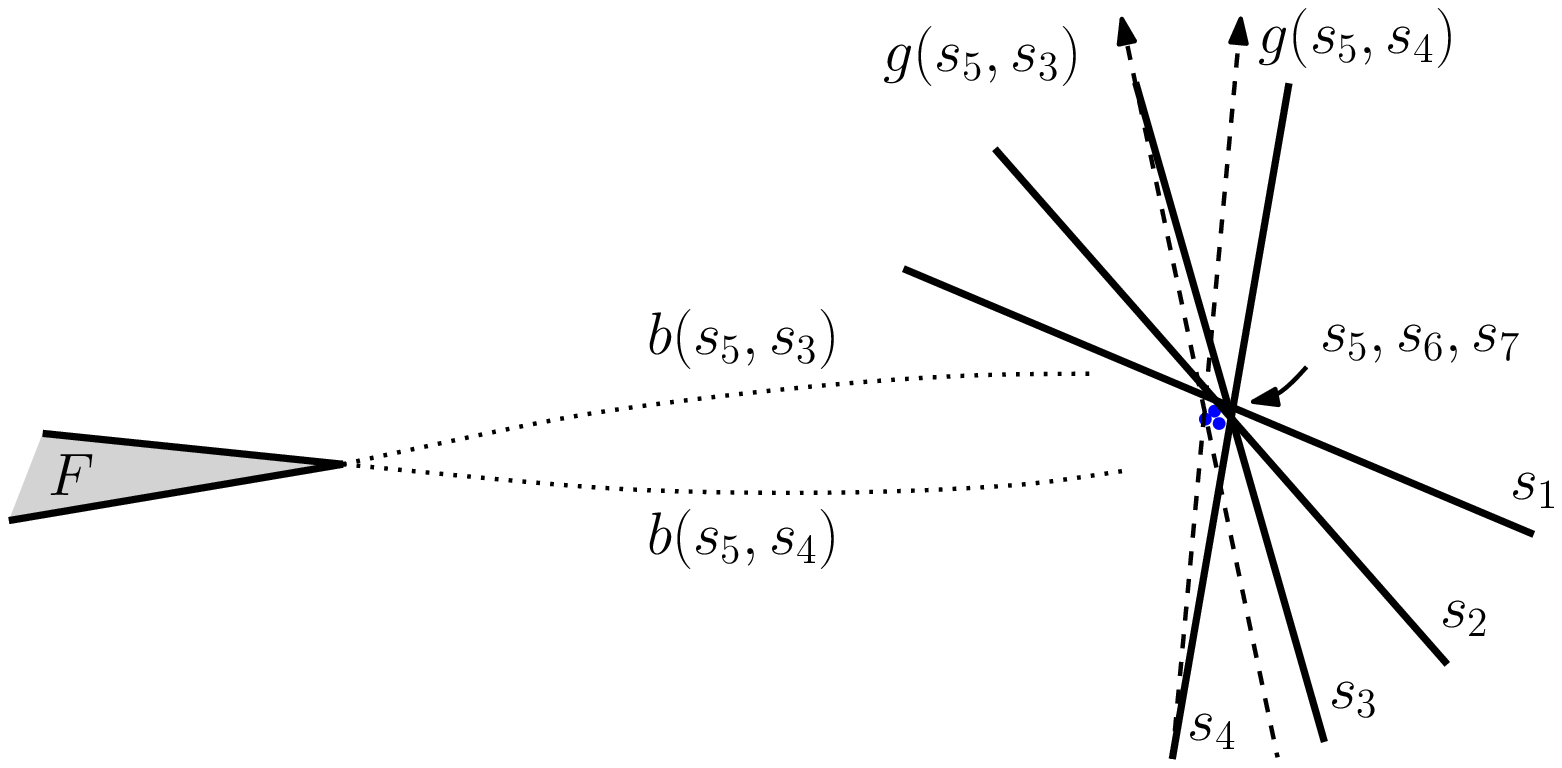}}
  \end{center}
  \caption{(a) During the rotation of the directed line, the positions in which the open halfplane to the left of it intersects all non-degenerate segments, alternate with the positions in which it does not; (b) The alternations produce bisectors that bound distinct unbounded faces of the region $V_4(H,S)$.}
  \label{fig:disc_segm02}
\end{figure}
For small $k$, $1<k<n/2$, the 
number of faces in the first example ($n-k+1$) is $\Omega(n)$, while for  large $k$, $n/2\leq k\leq n-1$, the number of faces in the second example $k$ is also $\Omega(n)$. 
\end{proof}

It may seem as if disconnected regions are present because of the crossings between segments, however, this is not the case.
In the example of Fig.~\ref{fig:disc_segm02}, we can untangle the segments to form a non-crossing configuration, while the same phenomena remain. 
Consider a segment  $s\in H$ whose endpoints define two supporting halfplanes.
We can move the endpoints of $s$ along the boundaries of the halfplanes away from the rest of the line segments in $H$, and untangle all line segments in $H$, while maintaining the same halfplanes that define the corresponding unbounded Voronoi edges.
For $k=n-1$, this was illustrated in~\cite{Aurenhammer06}.

\begin{lemma}
  An order-$k$ region $V_k(H,S)$ has $O(k)$ unbounded disconnected faces.
\end{lemma}
\begin{proof} 
  We show that an endpoint $p$ of a segment $s\in H$ may induce at most two unbounded Voronoi edges bordering $V_k(H,S)$ (see Fig.~\ref{fig:disc_unb_faces01}).
  Consider two such unbounded Voronoi edges. 
  By Corollary~\ref{cor:unb_edge}, there are open halfplanes $h_1$, $h_2$, such that the boundary of $h_1$ and $h_2$ pass through point $p$ and the endpoints of the line segments $t_1$ and $t_2$, respectively.
  The open halfplanes $h_1$ and $h_2$  intersect all line segments in $H$ and do not intersect line segments in $S\setminus H$.
  Thus, any other supporting halfplane $h_3$, with boundary passing through point $p$ and an endpoint of some line segment $s_3\in S\setminus H$, must intersect either $t_1$ or $t_2$. 
  Since $|H|=k$ and a segment has two endpoints, the claim follows.
  \begin{figure}
    \begin{center}
      \includegraphics[width=1.6 in]{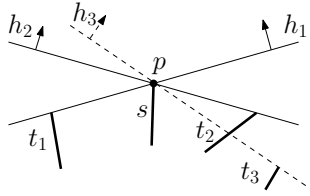}
    \end{center}
    \caption{Every endpoint of a segment $s\in H$ can induce at most two supporting halfplanes.}
    \label{fig:disc_unb_faces01}
  \end{figure}
\end{proof}

Although an order-$k$ Voronoi region may be disconnected, the union of all faces induced by a segment $s$ is a connected region which encloses $s$.
In particular, let $V_k(s,S)=\cup_{H\subset S, s\in H}\overline{V_k(H,S)}$, where $\overline{V_k(H,S)}$ denotes the closure of an order-$k$ region.
A set $X$ is said to be {\em weakly star-shaped}\/ with respect to a line segment $s$ if for every point $x\in X$ there exists a point $y\in s$, such that the line segment $xy$ is entirely enclosed in $X$.

\begin{lemma}\label{lemma:star_shaped}
  Consider the order-$k$ Voronoi diagram $\mathcal{V}_k(S)$.
  The union of all faces in $\mathcal{V}_k(S)$ affiliated with a segment $s$, $V_k(s,S)$, is weakly star-shaped with respect to $s$ ($s\in V_k(s,S)$). 
\end{lemma}
\begin{proof}
  Let $x$ be an arbitrary point in $V_k(s,S)$.
  Denote by $D_k(x)$ the minimum disk, centered at $x$, that intersects at least $k$ line segments, and by $D_s(x)$ the minimum disk, centered at $x$, that touches the line segment $s$.
  Since $x\in V_k(s,S)$, $x$ must be in one of the regions $V_k(H,S)$, where $s\in H$.
  Therefore, $D_s(x)\subseteq D_k(x)$.

  Let $y$ be a point on the line segment $s$ that is closest to $x$.
  Consider an arbitrary point $a$ on the line segment $xy$. 
  Then, $D_s(a)\subseteq D_s(x)\subseteq D_k(x)$.
  This implies that the line segment $s$ is the $i$th-closest line segment from point $a$, where $i\leq k$.
  Therefore, $a\in V_k(s,S)$.
  Since $a$ is taken arbitrarily, the entire line segment $xy$ is enclosed in $V_k(s,S)$. 
\end{proof}

\section{Structural Properties and Complexity}\label{section:structural_complexity}
In this section we show structural properties of the order-$k$ Voronoi diagram of $n$ disjoint line segments and  prove that its  combinatorial complexity is $O(k(n-k))$, despite the presence of disconnected regions. 

We first prove Theorem~\ref{lemma:lee}, which is a generalization to line segments of the formula in~\cite[Theorem~2]{Lee82}, which counts the total number of faces of $\mathcal{V}_k(S)$ as a function of $n,k$, and number of unbounded edges.
To this aim, we exploit the fact that the farthest line-segment Voronoi diagram is a tree structure~\cite{Aurenhammer06}. 
Then in Lemma~\ref{lemma:unb}, we analyze the number of unbounded edges in the order-$k$ Voronoi diagram in a dual setting by using results on arrangements of wedges~\cite{Aurenhammer06,Edelsbrunner82} and $\kleq$-level in arrangements of Jordan curves~\cite{Sharir95}. 
We derive the result by combining Theorem~\ref{lemma:lee} and Lemma~\ref{lemma:unb}.

The definition of an order-$k$  Voronoi region implies that two adjacent order-$k$ Voronoi faces must differ in exactly two sites. 
Therefore, any point on a Voronoi edge separating two faces, must be  the center of a disk that intersects $k{+}1$ and touches two line segments. 
Under the general position assumption, an order-$k$ Voronoi vertex $v$ is incident to three Voronoi edges and to three faces. 
Thus, order-$k$ Voronoi regions can be of two types~\cite{Lee82}:
(1) $V_k(H\cup\{a\},S)$, $V_k(H\cup\{b\},S)$, $V_k(H\cup\{c\},S)$; or (2) $V_k(H\cup\{a,b\},S)$, $V_k(H\cup\{b,c\},S)$, $V_k(H\cup\{c,a\},S)$.
In the first case, $|H|=k-1$ and $v$ is called a {\em new}\/ order-$k$ Voronoi vertex.
In the second case, $|H|=k-2$ and $v$ is called an {\em old}\/ order-$k$ Voronoi vertex.
In both cases, $v$ is the center of the disk whose interior intersects all the line segments in $H$, and whose boundary touches the line segments $a,b$ and $c$.
Thus, Voronoi vertices in $\mathcal{V}_k(S)$ are classified into {\em new}\/ and {\em old}. 
A new Voronoi vertex in $\mathcal{V}_k(S)$ is an old Voronoi vertex in $\mathcal{V}_{k+1}(S)$, and it appears for the first time in the  order-$k$ diagram.
Under the general position assumption, an {\em old}\/ Voronoi vertex in $\mathcal{V}_k(S)$ is a {\em new}\/ Voronoi vertex in $\mathcal{V}_{k-1}(S)$.

\begin{lemma}\label{lemma:face_tree}
  Consider a face $F$ of the region $V_{k+1}(H,S)$ ($|H|=k+1$). 
  The portion of $\mathcal{V}_k(S)$ enclosed in $F$ is exactly the portion of the farthest Voronoi diagram $\mathcal{V}_f(H)$ enclosed in $F$.
\end{lemma}
\begin{proof}

  Let $x$ be a point in $F$.  
  Suppose $x$ belongs to the region $V_k(H_j,S)$ of $\mathcal{V}_k(S)$. 
  Then $H_j$ is the set consisting of the $k$ line segments closest to $x$.
  Let $\{s_j\}=H\setminus H_j$; then $s_j$ is the $k{+}1$-closest line segment to $x$.
  Therefore, $s_j$ is the line segment farthest from $x$, among all segments in $H$.
  Therefore, $x\in V_f(s_j,H)$.

  Suppose $x$ belongs to the edge separating regions $V_k(H_j,S)$ and $V_k(H_r,S)$ of $\mathcal{V}_k(S)$.
  Then we can show in a similar way that $x$ belongs to the edge separating farthest regions $V_f(s_j,H)$ and $V_f(s_r,H)$, where $\{s_j\}=H\setminus H_j$ and $\{s_r\}=H\setminus H_r$.
\end{proof}

Consider a region $V_f(s,H)$ of the farthest line-segment Voronoi diagram, $\mathcal{V}_f(H)$.
This region has the following {\em visibility property}, 
 see Fig.~\ref{fig:KxKy}.

\begin{lemma}[Visibility 
property in a farthest Voronoi region]\label{lemma:visibility_property}
  Let $x$ be a point in a farthest Voronoi region $V_f(s,H)$ of $\mathcal{V}_f(H)$.
  Let $r(s,x)$ be the ray realizing the distance $d(s,x)$, emanating from point $p\in s$ such that $d(p,x)=d(s,x)$, and extending to infinity.
  The ray $r(s,x)$ must intersect the boundary of $V_f(s,H)$ at a point $a_x$, and the unbounded portion of $r(s,x)$ beyond $a_x$ must lie entirely in $V_f(s,H)$.  
\end{lemma}
\begin{proof}
  Consider a point $y$ along $r(s,x)$, which is a slight translation of the point $x$ towards $p$.
  Let $D_x$ (resp., $D_y$) be the minimum disk centered at $x$ (resp., $y$), that intersects all segments in $H$.
  Then, $D_y\subset D_x$.
  The disk $D_y$ intersects all segments in $H$ and touches $s$ at point $p$, which implies that $y\in V_f(s,H)$.
  If we continue to move $y$ towards $p$, the disk $D_y$ will eventually touch some segment in $H\setminus\{s\}$, at position $y=a_x$.
  Therefore, the point $a_x$ belongs to an edge of the farthest Voronoi diagram $\mathcal{V}_f(H)$.
  Now, if we move $y$, starting from $x$ and away from $p$, then the disk $D_y$ will continue to contain $D_x$ and touch $s$.
  Therefore, the part of the ray $r(s,x)$ beyond $a_x$ must entirely belong to $V_f(s,H)$.
\end{proof}

Using this property, we derive the following lemma.

\begin{lemma}\label{lemma:tree}

  Let $F$ be a face of a region $V_{k+1}(H,S)$ in $\mathcal{V}_{k+1}(S)$. 
  The graph structure of $\mathcal{V}_k(S)$ enclosed in $F$ is a tree that consists of at least one edge.
  Each leaf of the tree is incident to an old Voronoi vertex on the boundary of $F$ (see Fig.~\ref{fig:KxKy}).
\end{lemma}
\begin{figure}[h!]
  \begin{center}
    \includegraphics[width=2.5in, clip, trim = 0mm 5mm 13mm 16mm]{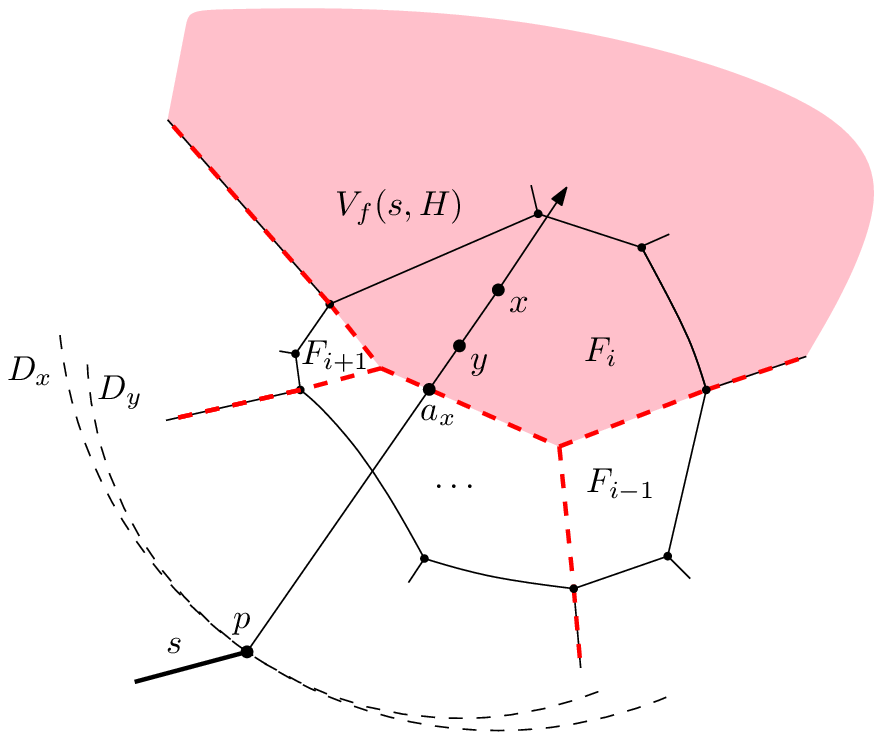}
  \end{center}
  \caption{The part of the ray $r(s,x)$ beyond $a_x$ entirely belongs to $V_f(s,H)$.} 
  \label{fig:KxKy}
\end{figure}

\begin{proof}
  Consider a point $x$ in $F$ (see Fig.~\ref{fig:KxKy}) and let $s$ be the segment in $H$ farthest away from $x$.
  Consider the ray $r(s,x)$, and the point $a_x$ as defined in Lemma~\ref{lemma:visibility_property}.
  Lemma~\ref{lemma:visibility_property} implies that $a_x$ is a point in the interior of $F$, therefore, $F$ must contain a portion of the tree of $\mathcal{V}_f(H)$, and, thus Lemma~\ref{lemma:face_tree} implies that  $F$ must contain at least one edge of $\mathcal{V}_k(S)$.

  Now we prove that the portion of $\mathcal{V}_k(S)$ enclosed in $F$ is connected.  
  Lemma~\ref{lemma:face_tree} implies that this portion is equal to the portion of $\mathcal{V}_f(H)$ enclosed in $F$.
  Assume, to the contrary, that the portion of $\mathcal{V}_f(H)$ enclosed in $F$ is disconnected.  
  Then, there is a subface $F_i$ of $F$ that separates two disconnected subtrees of $\mathcal{V}_f(H)$, say, $T_1$ and $T_2$.  
  Let $F_i\subseteq V_f(s,H)$, $v$ be a point on the boundary  of $V_f(s,H)$ between $T_1$ and $T_2$, and $r(s,v)$ be the ray that realizes the distance from $s$ to $v$ extending to infinity.  
  The {\em visibility property}\/ of  $V_f(s,H)$ in Lemma~\ref{lemma:visibility_property} implies that the portion of $r(s,v)$ beyond $v$ belongs  entirely to $V_f(s,H)$.  
  Since $T_1$ and $T_2$ bound $F_i$, the ray $r(s,v)$ must intersect $F_i$ beyond the point $v$.
  Consider the minimum disk centered at $v$, that intersects all segments in $H$.  
  The disk must also intersect some segments in $S\setminus H$ because $v$ does not belong to $F$.  
  If we move the center of the disk along $r(s,v)$ away from $s$, the new minimum disk will contain the previous disk, and, therefore, it will also intersect the same segments in $S\setminus H$.
  Thus, no portion of $r(s,v)$ can be in $F$, which is a contradiction.
\end{proof}

\begin{corollary}\label{corollary:tree}
  Consider a face $F$ of the Voronoi region $V_{k+1}(H,S)$. 
  Let $m$ be the number of Voronoi vertices in the portion of $\mathcal{V}_k(S)$ enclosed in the interior of $F$. 
  Then, $F$ encloses $2m{+}1$ Voronoi edges of $\mathcal{V}_k(S)$.
\end{corollary}

Let $F_k$, $E_k$, $V_k$, and $U_k$ denote the number of faces, edges, vertices and edges faces in $\mathcal{V}_k(S)$  respectively.
If an edge is unbounded in both directions, then it is counted twice.
By the Euler's formula  we derive the following lemma.
\begin{lemma}\label{lemma:euler}
  \begin{eqnarray}
    E_k=3(F_k-1)-U_k\label{eq:euler1}\\
    V_k=2(F_k-1)-U_k\label{eq:euler2}
  \end{eqnarray}
\end{lemma}
\begin{proof}
  Consider $\mathcal{V}_k(S)$ and connect every unbounded edge with an artificial point at infinity.
  Then Euler's formula implies that $F_k-E_k+V_k=1$.

  Consider the dual graph of $\mathcal{V}_k(S)$. 
  Connect every vertex of the dual graph, representing an unbounded face of $\mathcal{V}_k(S)$, with an artificial point at infinity.
  If an unbounded face is incident to four unbounded edges, then connect the corresponding vertex twice.
  Then, under the general-position assumption, every face in the dual graph must have exactly three edges, and every edge is adjacent to exactly two faces.
  Therefore, $3(V_k+U_k)=2(E_k+U_k)$. 
  The combination of these equations proves the lemma. 
\end{proof}

\begin{lemma}\label{lemma:total_unb}
  The total number of unbounded edges in the order-$k$ Voronoi diagram of all orders is $\sum_{i=1}^{n-1}U_i=n(n-1)$.
\end{lemma}
\begin{proof}
  Consider an arbitrary  pair of segments $s_1$ and $s_2$.
  There are exactly two open halfplanes $r_1$ and $r_2$ that touch $s_1$ and $s_2$. 
  Corollary~\ref{cor:unb_edge} implies that these open halfplanes define unbounded Voronoi edges for some order-$(k_1{+}1)$ and order-$(k_2{+}1)$ Voronoi diagrams, where $k_1$ and $k_2$ are the numbers of segments that $r_1$ and $r_2$ intersect, respectively. 
  In addition, any unbounded Voronoi edge is induced by such a halfplane.
  Thus, $\sum_{i=1}^{n-1}U_i=2{n\choose 2}=n(n-1)$.       
\end{proof}

\begin{theorem}\label{lemma:lee}
  The number of faces in the order-$k$ Voronoi diagram of $n$ disjoint line segments is
  \begin{eqnarray}
    F_k=2kn-k^2-n+1-\sum_{i=1}^{k-1}U_i \label{eq:app_lee}\\
    \mbox{or, equivalently,}\quad F_k=1-(n-k)^2+\sum_{i=k}^{n-1}U_i \label{eq:app_dual_lee}
  \end{eqnarray}
\end{theorem}
\begin{proof}
  Let $V_k$, $V'_k$ and $V''_k$ be the number of Voronoi vertices, {\em new}\/ Voronoi vertices, and {\em old}\/ Voronoi vertices in $\mathcal{V}_k(S)$, respectively.  (Notation follows~\cite{Lee82}.)
  Then, $V_k=V'_k+V''_k=V'_k+V'_{k-1}$.

  Following~\cite{Lee82}, we obtain a recursive formula for the number of faces $F_k$ of the order-$k$ Voronoi diagram.
  Assuming that segments do not intersect, $F_1=n$, since each segment induces exactly one face in $\mathcal{V}_1(S)$.
  In  $\mathcal{V}_2(S)$, each face encloses exactly one edge of $\mathcal{V}_1(S)$, thus, $F_2=E_1$. 
  Then by  Lemma~\ref{lemma:euler} we derive $F_2=3(F_1-1)-U_1$, thus, $F_2=3(n-1)-U_1$.

  We now prove that $F_{k+2}=E_{k+1}-2V'_k$ (\emph{Claim 1}).  
  Note that $V'_1=V_1$ and $V_1=2(n-1)-U_1$ (using Eq.~(\ref{eq:euler2}) of Lemma~\ref{lemma:euler}).
  The definition of {\em old}\/ Voronoi vertices implies that {\em old}\/ Voronoi vertices of $\mathcal{V}_{k+1}(S)$ lie in the interior of the faces of  $\mathcal{V}_{k+2}(S)$. 
  Consider a face $F_i$ of  $\mathcal{V}_{k+2}(S)$.  
  Let $m_i$ be the number of {\em old}\/ Voronoi vertices of  $\mathcal{V}_{k+1}(S)$ enclosed in the interior of $F_i$.
  Then, $F_i$ encloses $e_i=2m_i+1$ Voronoi edges of  $\mathcal{V}_{k+1}(S)$  (see Corollary~\ref{corollary:tree}).
  Summing up the numbers of all faces in $\mathcal{V}_{k+2}(S)$, we obtain that $\sum_{i=1}^{F_{k+2}}e_i=2\sum_{i=1}^{F_{k+2}}m_i+F_{k+2}$.
  However, $\sum_{i=1}^{F_{k+2}}m_i=V''_{k+1}=V'_k$ and $\sum_{i=1}^{F_{k+2}}e_i=E_{k+1}$.
  Therefore, $F_{k+2}=E_{k+1}-2V'_k$, and, Claim 1 follows.

  We now use Claim 1 to obtain a recursive formula for $F_k$.
  Summing up $F_{k+2}$ and $F_{k+3}$, we obtain $F_{k+3}=E_{k+2}+E_{k+1}-F_{k+2}-2V'_{k+1}-2V'_k=E_{k+2}+E_{k+1}-F_{k+2}-2V_{k+1}$.
  We then substitute Eqs.~(\ref{eq:euler1}) and (\ref{eq:euler2}) in the last formula and obtain 
  \begin{equation}
    F_{k+3}=2F_{k+2}-F_{k+1}-2-U_{k+2}+U_{k+1}.\label{eq:iterative}
  \end{equation}
  where, $F_1=n$ and  $F_2=3(n-1)-U_1$. 
  Because $F_2=E_1$, Eq.~(\ref{eq:iterative}) can also be derived for $F_3$, i.e.\ the formula applies to $k\geq 0$. 

  By induction, using Eq.~(\ref{eq:iterative}) and the above base cases, we derive Eq.~(\ref{eq:app_lee}). 
  Lemma~\ref{lemma:total_unb} implies that $\sum_{i=1}^{k-1}U_i+\sum_{i=k}^{n-1}U_i=\sum_{i=1}^{n-1}U_i=n(n-1)$.
  Combining this result with Eq.~(\ref{eq:app_lee}), we derive Eq.~(\ref{eq:app_dual_lee}).
\end{proof}

\begin{lemma}\label{lemma:unb}
  Given a set $S$ of $n$  line segments, $\sum_{i=k}^{n-1}U_i=O(n(n-k))$.
\end{lemma}
\begin{proof}
  We use the well-known point-line duality transformation $T$ that maps a point $p=(a,b)$ in the primal plane to a line $T(p): y=ax-b$ in the dual plane, and vice versa (see~\cite{Aurenhammer06}). 
  We call the set of points above both lines $T(p)$ and $T(q)$ the {\em wedge}\/ of $s=(p,q)$. 
  Consider a line $\ell$ and a segment $s=(p,q)$. 
  The segment $s$ is above the line $\ell$ if and only if the point $T(\ell)$ is strictly above lines $T(p)$ and $T(q)$~\cite{Aurenhammer06}.  

  Consider the arrangement $W$ of the wedges $w_i$, $i=1,\ldots,n$, corresponding to the segments in $S=\{s_1,\ldots,s_n\}$. 
  For our analysis we need the notions of $r$-level and $\rleq$-level. 
  The $r$-level of $W$ is the set of edges such that every point along an edge lies above $r$ wedges. 
  The $r$-level shares its vertices with the $(r{-}1)$-level and the $(r{+}1)$-level. 
  The $\rleq$-level of $W$ is the  set of edges such that every point on it is above at most $r$ wedges. 
  For our purposes in this paper, the complexity of the $r$-level and the $\rleq$-level is the number of their vertices, excluding the wedge apices. 
  We denote the maximum complexity of the $r$-level and the $\rleq$-level of $n$ wedges by $g_r(n)$ and $g_{\leq r}(n)$, respectively. 
  We first prove the following claim.
  \begin{claim}
    \topsep=0pt
    The number of unbounded Voronoi edges of $V_k(S)$, unbounded in direction $\phi\in[\pi,2\pi]$, is exactly the number of vertices shared by the $(n{-}k{-}1)$-level and the $(n{-}k)$-level of $W$.
    Thus, $U_k=O(g_{n-k-1}(n))$.
  \end{claim}
  \noindent
  {\emph Proof of Claim.}
  Let $s_i, s_j$ be two line segments that define an unbounded bisector in a direction $\phi\in[\pi,2\pi]$.
  Then, there is a line $\ell$ passing through their endpoints, such that the open halfplane $\ell^-$ below $\ell$ intersects $k{-}1$ line segments and does not intersect $s_i$ nor $s_j$.
  Then, $\ell$ passes strictly below $n-(k-1)-2=n-k-1$ line segments.
  Thus, $\ell$ corresponds to a point $p$ in the arrangement of wedges shared by the $(n{-}k{-}1)$-level and $(n{-}k)$-level (see Fig.~\ref{fig:wedges01}).
  By the above claim
  \begin{equation}
    \sum_{i=k}^{n-1}U_i=O(g_{\leq n-k-1}(n)).\label{f03}
  \end{equation}
  Since the arrangement of wedges is a special case of arrangements of Jordan curves, we  use a formula from~\cite{Sharir95} to bound the complexity of the $\rleq$-level in such an arrangement:
  \begin{equation}
    g_{\leq r}(n)=O\left((r+1)^2g_0\left(\left\lfloor \frac{n}{r+1}\right\rfloor\right)\right)\label{f04}
  \end{equation}
  The complexity of the lower envelope of such wedges $g_0(n)$ is $O(n)$~\cite{Aurenhammer06,Edelsbrunner82}.
  (In~\cite{Sharir95} one can find the weaker bound $g_0(n)=O(n\log{n})$).
  Therefore, $g_{\leq r}(n)=O(n(r+1))$.
  By substituting this into Eq.~(\ref{f03}) we obtain that $\sum_{i=k}^{n-1}U_i=O(n(n-k))$.

  \begin{figure}[!h]
    \begin{center}
      \includegraphics[width=3.0 in, clip, trim = 0mm 0mm 0mm 0mm]{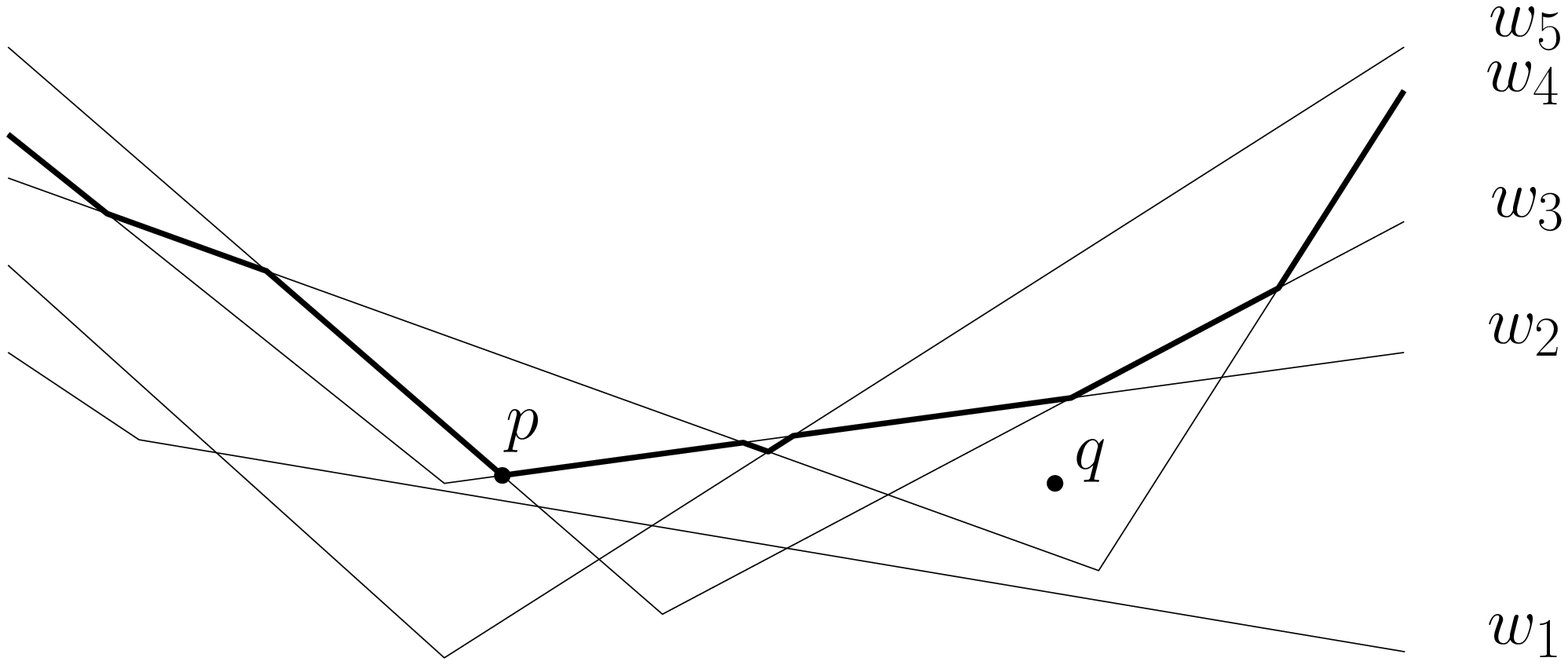}
      \\\vspace{2mm}
      \includegraphics[width=3.0 in, clip, trim = 0mm 0mm 0mm 0mm]{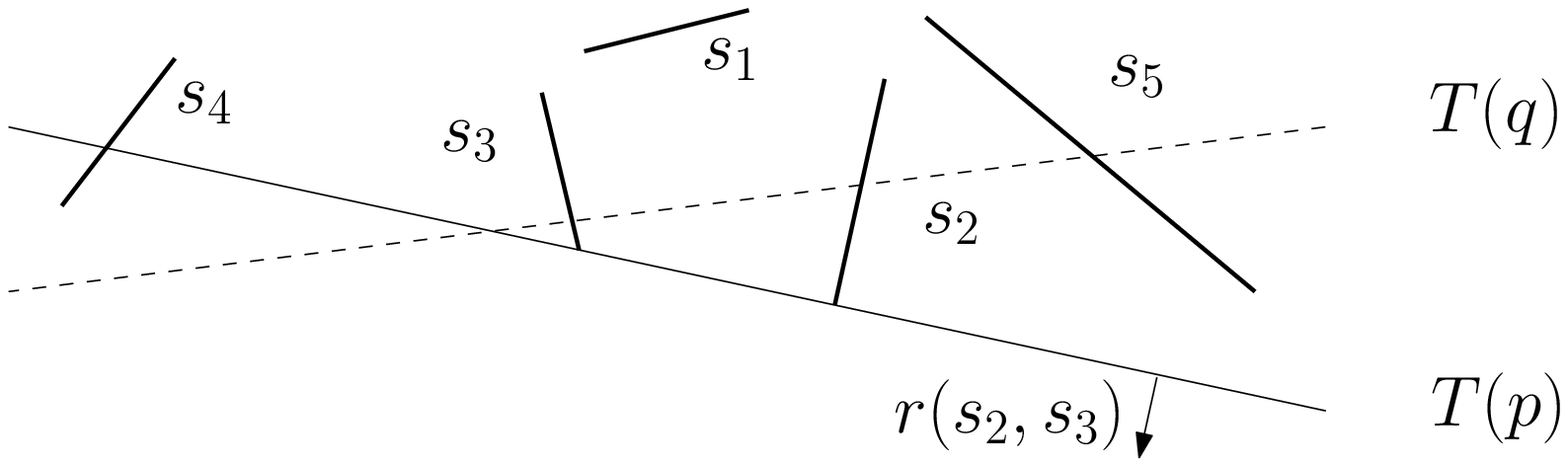}
    \end{center}
    \caption{(a) In the dual plane, the point $p$ belongs to the 2-level and the 3-level of the arrangement $W$; 
      (b) In the primal plane, the halfplane $r(s_2,s_3)$ below $T(p)$ defines the unbounded Voronoi edge that separates $V_2(\{s_2,s_4\},S)$ and $V_2(\{s_3,s_4\},S)$.}
      \label{fig:wedges01}
    \end{figure}
  \end{proof}
  By combining  Lemma~\ref{lemma:unb} and Theorem~\ref{lemma:lee}, we obtain the following result. 
  \begin{theorem} \label{theorem:bound} 
    The combinatorial complexity of the order-$k$ Voronoi diagram of $n$ disjoint line segments is $F_k=O(k(n-k))$
  \end{theorem}
  \begin{proof}
    For $1\leq k<n/2$, Eq.~(\ref{eq:app_lee}) implies that $F_k=O(k(n-k))$.

    For $n/2\leq k\leq n-1$, Lemma~\ref{lemma:unb} implies that $\sum_{i=k}^{n-1}U_i=O(n(n-k))=O(k(n-k))$.
    The dual formula~(\ref{eq:app_dual_lee}) implies that $F_k=1-(n-k)^2+\sum_{i=k}^{n-1}U_i\leq \sum_{i=k}^{n-1}U_i$, which is $O(k(n-k))$.
  \end{proof}

  \section{Segments forming a Planar Straight-Line Graph}\label{section:abutting_segments}
  In this section we consider line segments that may touch at endpoints, such as line segments forming a simple polygon, more generally, line segments forming a \emph{planar straight-line graph} (\emph{PSLG}, in short).
  This is important for applications that involve polygonal objects in the plane, for an example  see~\cite{Papadopoulou11}.

  Line segments forming a PSLG are inherently degenerate because of vertices in the PSLG of degree greater than one.
  These vertices induce areas on the plane that are equidistant from multiple segments, whose number is independent of $k$.  
  The problem remains, even under a {\em weak general position assumption}\/ that no more than three {\em elementary sites}\/ touch the same circle. 
  A segment consists of three elementary sites: two endpoints and an open line segment.\footnote{In case the line through a segment $s$ is tangent to a circle $C$ at one of the segment endpoints, both elementary sites, the endpoint and the open portion of $s$, touch $C$.  
  Otherwise, only one elementary site can touch $C$.}
  Note that a PSLG cannot satisfy the standard general position assumption that no more than three sites can touch the same circle. 
  In terms of bisectors, degeneracies involving a PSLG manifest themselves in two ways:
  (1) bisectors that contain two-dimensional regions, such as the shaded area in Fig.~\ref{fig:pslg_degen}(a); 
  and (2) bisectors which intersect non-transversely, such as those illustrated in Fig.~\ref{fig:pslg_degen}(b); 
  as well as the combination of (1) and (2). 

  \begin{figure}[h]
    \includegraphics{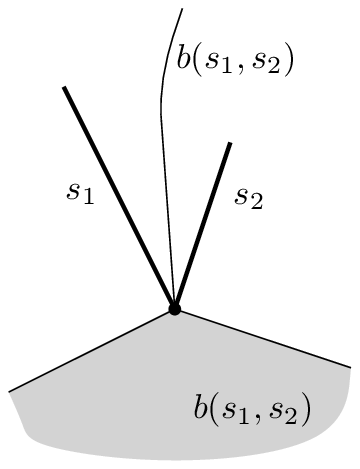}
    \includegraphics{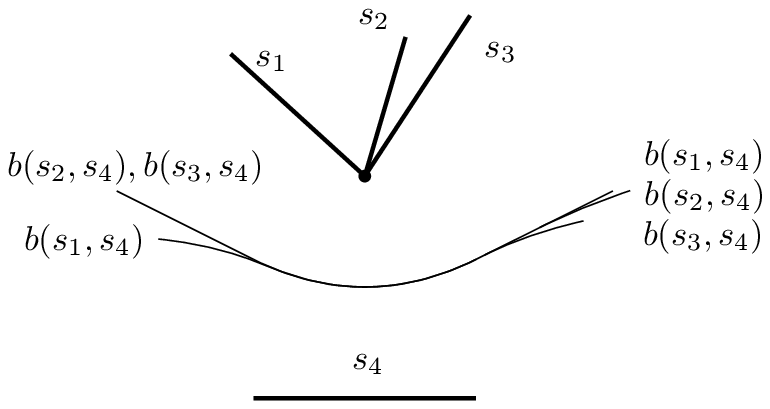}
    \\
    \put(45,0){(a)}
    \put(200,0){(b)}
    \caption{(a) A bisector containing a 2-dimensional portion; (b) Bisectors intersecting non-transversely.} 
    \label{fig:pslg_degen}
  \end{figure}

  For $k=1$, a standard convention to cope with the high-degree vertices of a PSLG, is to consider elementary sites as distinct entities, see e.g.,~\cite{Kirkpatrick79}.
  For $k>1$, this standard convention is not adequate because it alters the essence of an order-$k$ Voronoi diagram, especially in the case of disjoint line segments.
  For example, for  $k=n-1$, the farthest Voronoi diagram of the elementary sites is the farthest-point Voronoi diagram of the segment endpoints, and not the farthest line-segment Voronoi diagram as defined in~\cite{Aurenhammer06}.
  In addition, considering elementary sites as distinct, it does not resolve the issue of multiple equidistant elementary sites, whose number is independent of $k$.
  Similarly, this issue is not addressed by  other  standard techniques that deal with two-dimensional bisectors, such as assigning a priority to sites while offering an entire equidistant region to the segment of higher priority~\cite{Klein09}, or using an angular bisector to split equidistant regions~\cite{Aurenhammer06}. 
  Perturbation techniques (see e.g.,~\cite{Seidel98}) to transform the PSLG into a set of disjoint line segment, on the other hand, may create artificial faces and tedious decompositions that are unrelated to the problem under consideration,  see e.g., Fig.~\ref{fig:pert}(b).

  In the following, we augment the definition of an order-$k$ Voronoi diagram to address the phenomenon of areas with multiple equidistant sites, whose number is independent of $k$.

  \subsection{Augmenting the definition of an order-$k$ Voronoi region}

  \begin{definition}
    Let $D_k(x)$ be the disk of minimum radius, centered at point $x$, that intersects (or touches) at least $k$ line segments.
    $D_k(x)$ is called an order-$k$ disk.
    The set of line segments in $S$ that have a non-empty intersection with $D_k(x)$ is denoted as $S_k(x)$.
    If $D_k(x)$ touches exactly one elementary site $p$ then it is called a {\em proper order-$k$ disk}\/ and it is denoted as $D_k^p(x)$.
  \end{definition}

  \begin{figure}
    \includegraphics[width=2.3 in, clip, trim = 18mm 11mm 30mm 70mm]{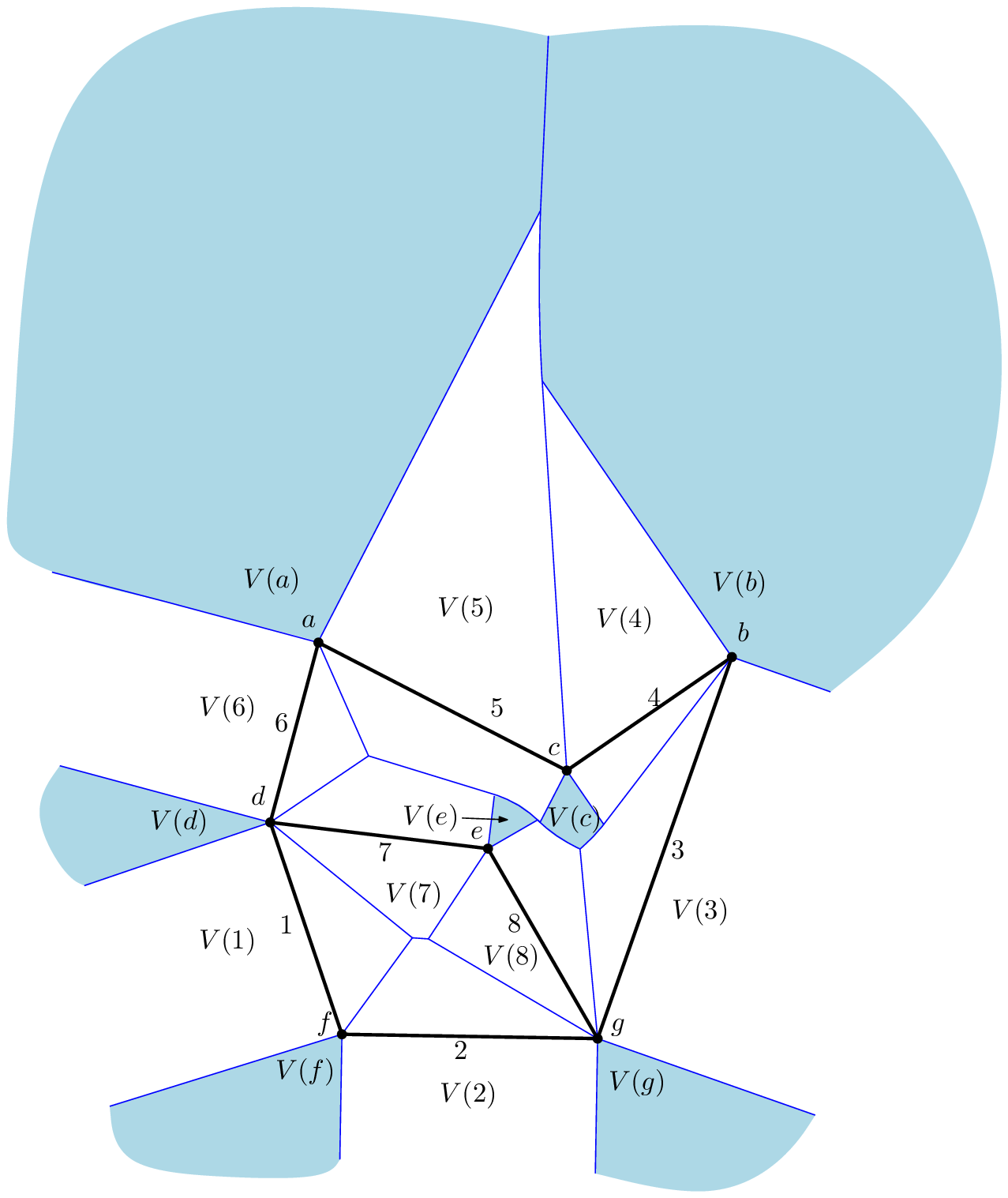}
    \includegraphics[width=2.3 in, clip, trim = 6mm 22mm 28mm 70mm]{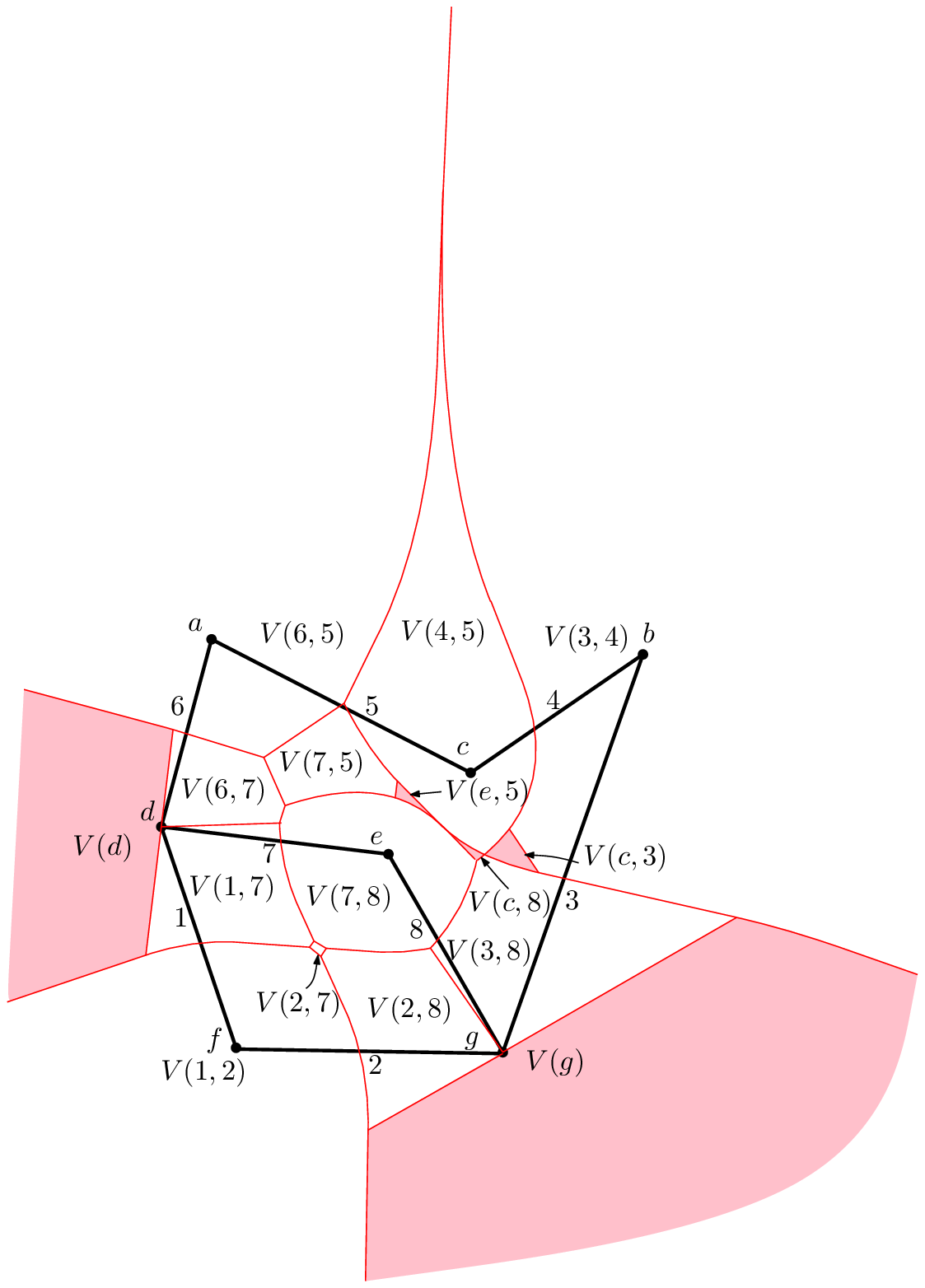}
    \\
    \put(80,0){(a)}
    \put(250,0){(b)} 
    \caption{$\mathcal{V}_1(S)$ for a PSLG. (a) $\mathcal{V}_1(S)$; (b) $\mathcal{V}_2(S)$.
Type-2 regions are shown shaded.
$V(e_1,\cdots,e_m)$ stands for Type-1 $V_k(H,S)$, $H=\{e_1,\dots,e_m\}$. $V(p,e_1,\cdots,e_m)$, stands for Type-2 $V_k(H_p,S)$, $\{e_1,\dots,e_m\}=H_p\setminus I(p)$.
}
    \label{fig:hovd_pslg}
  \end{figure}

  Clearly, $D_k(x)$, and thus, $S_k(x)$, are  unique for every point $x$ in the plane.
  Furthermore, if $D_k(x)$ is non-proper then $x$ must be a point along the bisector of  two elementary sites.

  \begin{definition}
    A set  $H\subseteq S$ is called an order-$k$ subset if
\begin{enumerate}
\item $|H|=k$ (Type-1); or
\item $|H|>k$ (Type-2), and there exists a proper order-$k$ disk $D_k^p(x)$, whose boundary passes through the common endpoint $p$ of at least two segments, and $S_k(x)=H$.
	Point $p$ is called a \emph{representative} of $H$. 
	An order-$k$ subset of representative $p$ is denoted as $H_p$. 
	The set of segments incident to $p$ is denoted as $I(p)$.
\end{enumerate}
  \end{definition}

  \paragraph{Remark.} A set of segments $H$ may have two (or more) representatives $p$, $q$, resulting in two distinct order-$k$ subsets $H_p$ and  $H_q$, 
where each has a distinct region in $\mathcal{V}_k(S)$.

  An order-$k$ Voronoi region can now be defined in terms of order-$k$ subsets of $S$ instead of cardinality-$k$ subsets.
  For a Type-1 subset $H$, we derive a Type-1 order-$k$ Voronoi region $V_k(H,S)$ as defined by Eq.~(\ref{eq:kregion}), which is equivalent to 
  $ V_k(H,S)=\{x~|~S_k(x) = H \}$. 
  For a Type-2 order-$k$ subset $H_p$ of representative $p$, we derive a Type-2 order-$k$ Voronoi region $V_k(H,S)$ defined as follows
  \begin{equation}\label{eq:def2}
    V_k(H_p,S)=\{x~|~S_k(x) = H_p \land D_k(x)=D_k^p(x)\} 
  \end{equation}

  Figure~\ref{fig:hovd_pslg} illustrates an example  of the 1st and 2nd order Voronoi diagram of a PSLG. 
  Type-2 Voronoi regions are illustrated shaded.

  The following lemma on Type-2 Voronoi regions is easy to derive following the definitions.

  \begin{lemma}\label{lemma:disk}
    Let  $V_k(H_p,S)$ be a Type-2 order-$k$ Voronoi region.
    Then    $\forall s\in H_p~\forall t\in S\setminus H_p$, $d(x,s)\leq d(x,p)<d(x,t)$, for any point $x$ in  $V_k(H_p,S)$. 
    Furthermore, $S_k(x)=S_{k+1}(x)$.
    $V_k(H_p,S)$ contains no graph elements of $\mathcal{V}_{k-1}(S)$ nor of $\mathcal{V}_{k+1}(S)$. 
  \end{lemma}

  \begin{figure}[h]
    \includegraphics[width=1.2 in, clip, trim = 0mm 0mm 0mm 3mm]{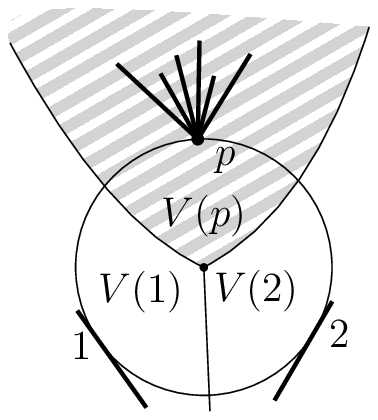}
    \includegraphics[width=1.0 in, clip, trim = 8mm -2mm 7mm 9mm]{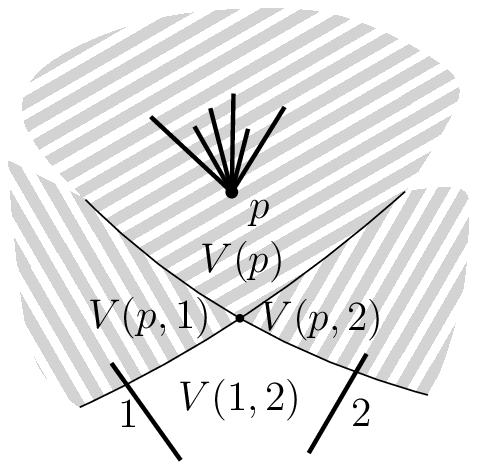}
    \includegraphics[width=1.0 in, clip, trim = 4mm 3mm 6mm 5mm]{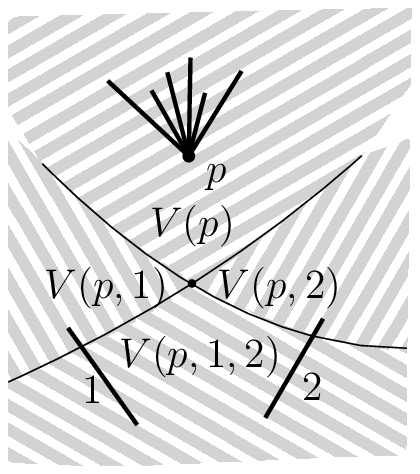}
    \includegraphics[width=1.0 in, clip, trim = 4mm 3mm 6mm 5mm]{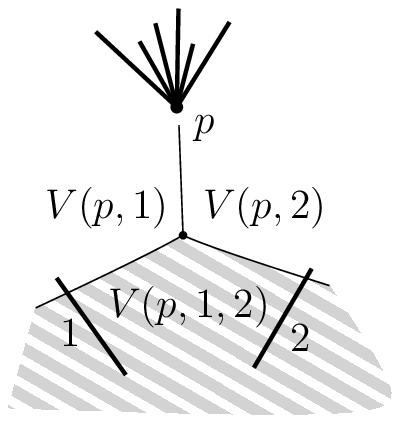}
    \\
    \put(45,0){(a)}
    \put(120,0){(b)}  
    \put(195,0){(c)}  
    \put(270,0){(d)}  
    \caption{A Type-2 Voronoi region of representative $p$, denoted as $V(p)$, and an incident Voronoi vertex for various orders $k$.
(a) $k=1$; (b) $k=2$; (c) $2<k\leq |I(p)|$ (for $k=|I(p)|$, 
$V(p)$ is Type-1); 
    (d) $k=|I(p)|+1$.
($V(p,e_1,\cdots,e_m)$ stands for $V_k(H_p,S)$, where $\{e_1,\dots,e_m\}=H_p\setminus I(p)$.)
} 
    \label{fig:pslg_vertices1}
  \end{figure}
  By Lemma~\ref{lemma:disk},  a Type-2 order-$k$ Voronoi region $V_k(H_p,S)$ can only enlarge in the order-$(k+1)$ diagram, spreading its influence into neighboring Type-1 regions. 
  At order $k=|H_p|$, $V_k(H_p,S)$ becomes Type-1.
  Figures~\ref{fig:pslg_vertices1} and~\ref{fig:pslg_vertices2} illustrate the evolution of a Type-2 region as the order of the diagram increases, where Type-2 regions are illustrated shaded.
  Figure~\ref{fig:pslg_vertices1}(a) depicts a vertex $v$ incident to a Type-2 region $V_1(H_p,S)$ (denoted for brevity in the figure as $V(p)$) and two Type-1 regions. 
  Figure~\ref{fig:pslg_vertices1}(b) shows how $p$ spreads into its neighboring regions and transforms them into Type-2 in $\mathcal{V}_{2}(S)$. 
  Figure~\ref{fig:pslg_vertices1}(c) shows the diagram for several orders $k$, $3\leq k\leq |I(p)|$.
  At $k=|I(p)|$, $V_k(H_p,S)$ becomes Type-1.
  Figure~\ref{fig:pslg_vertices1}(d) illustrates the diagram for order $k=|I(p)|+1$, when $V_k(H_p,S)$ has been absorbed by its Type-2 neighbors.
  Note that during this process, the degree of vertex $v$ is higher than three, despite the (weak) general-position assumption.

  Figure~\ref{fig:pslg_vertices2} illustrates an example of a vertex initially incident to three Type-2 Voronoi regions with representatives $p$, $r$, and $q$, respectively as shown in Fig.~\ref{fig:pslg_vertices2}(a).
  As the order increases, the Voronoi region of $q$ ($q$ has the smallest degree) becomes Type-1; 
in the next order it is split between two Type-2 regions of representatives $r$ and $p$ respectively, as shown in Fig.~\ref{fig:pslg_vertices2}(b).
  In Fig.~\ref{fig:pslg_vertices2}(c), after the region of $r$ ($V(r)$) becomes Type-1 for $k=|I(r)|$, it is split by the representatives of the neighboring Type-2 regions at order $k=|I(r)|+1$. 
  This creates a Voronoi vertex of degree five incident to portions of three bisectors.
  Later, the Voronoi region $V(p)$ will be split by its two neighbors and the incident  Voronoi vertex will obtain degree six.
  Under the weak general position assumption, six is the highest degree such a Voronoi vertex can obtain.

The order-$k$ subsets of two neighboring Type-2 Voronoi regions need not differ in exactly one element as in the ordinary case of Type-1 regions. In fact,  $V_k(H_p,S)$ and $V_k(J_p,S)$ may be neighboring and  $H_p\subset J_p$, as shown in Figures~\ref{fig:pslg_vertices2}(b),(c). 
  In this case, the Voronoi edge bounding the two regions is portion of the bisector  $b(p,y)$ for  $y\in J_p\setminus H_p$, but for any point $t\in V_k(J_p,S)$, $d(t,J_p)=d(t,p)$.
  Type-2 Voronoi regions illustrate the peculiarities listed above, however, they pose no difficulty in the construction of the diagram.
  The complexity of the diagram remains $O(k(n-k))$ as shown in the following subsection.

  \begin{figure}[h]
    \includegraphics[width=1.5 in, clip, trim = 15mm 15mm 29mm 15mm]{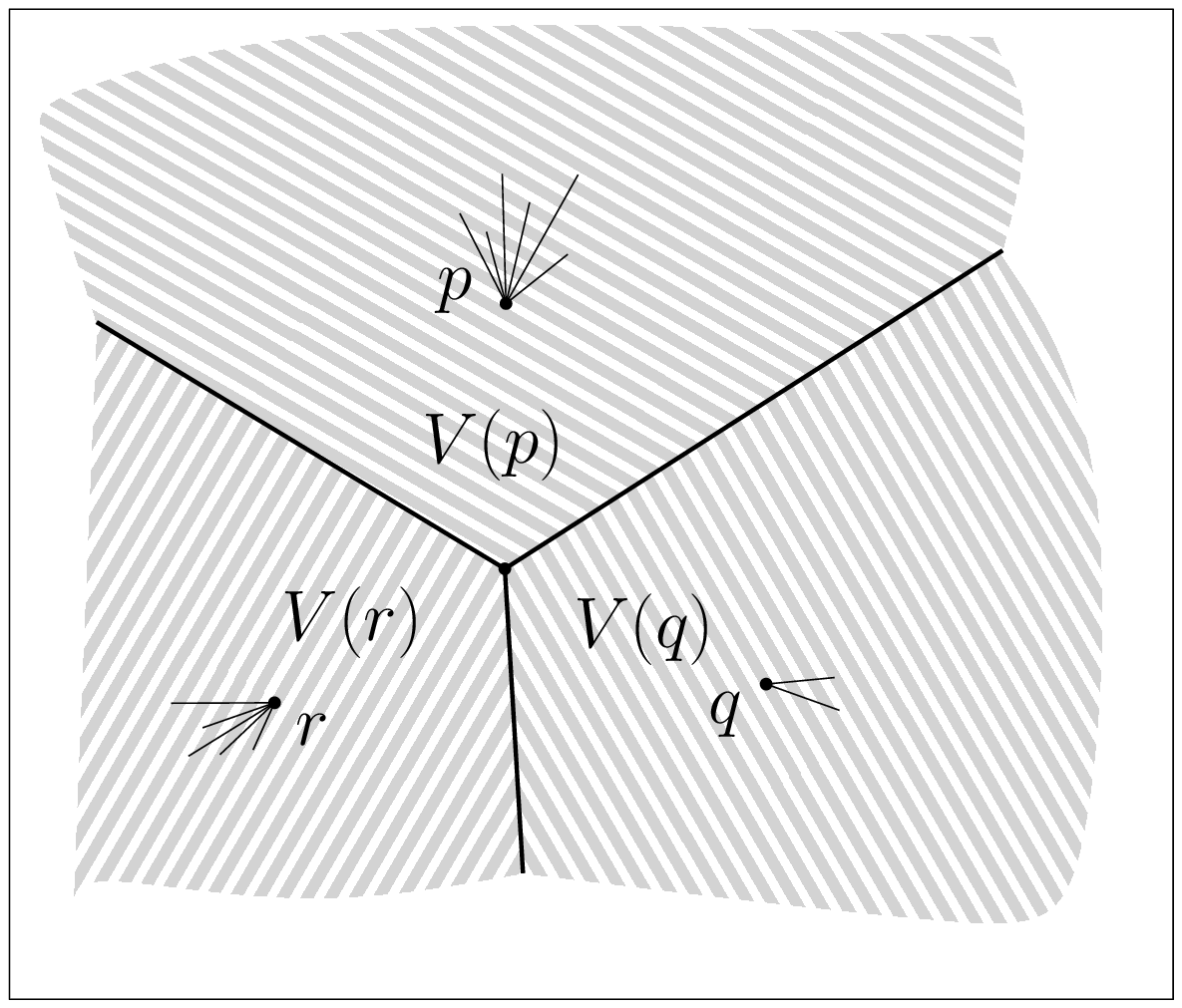}
    \includegraphics[width=1.5 in, clip, trim = 15mm 15mm 29mm 15mm]{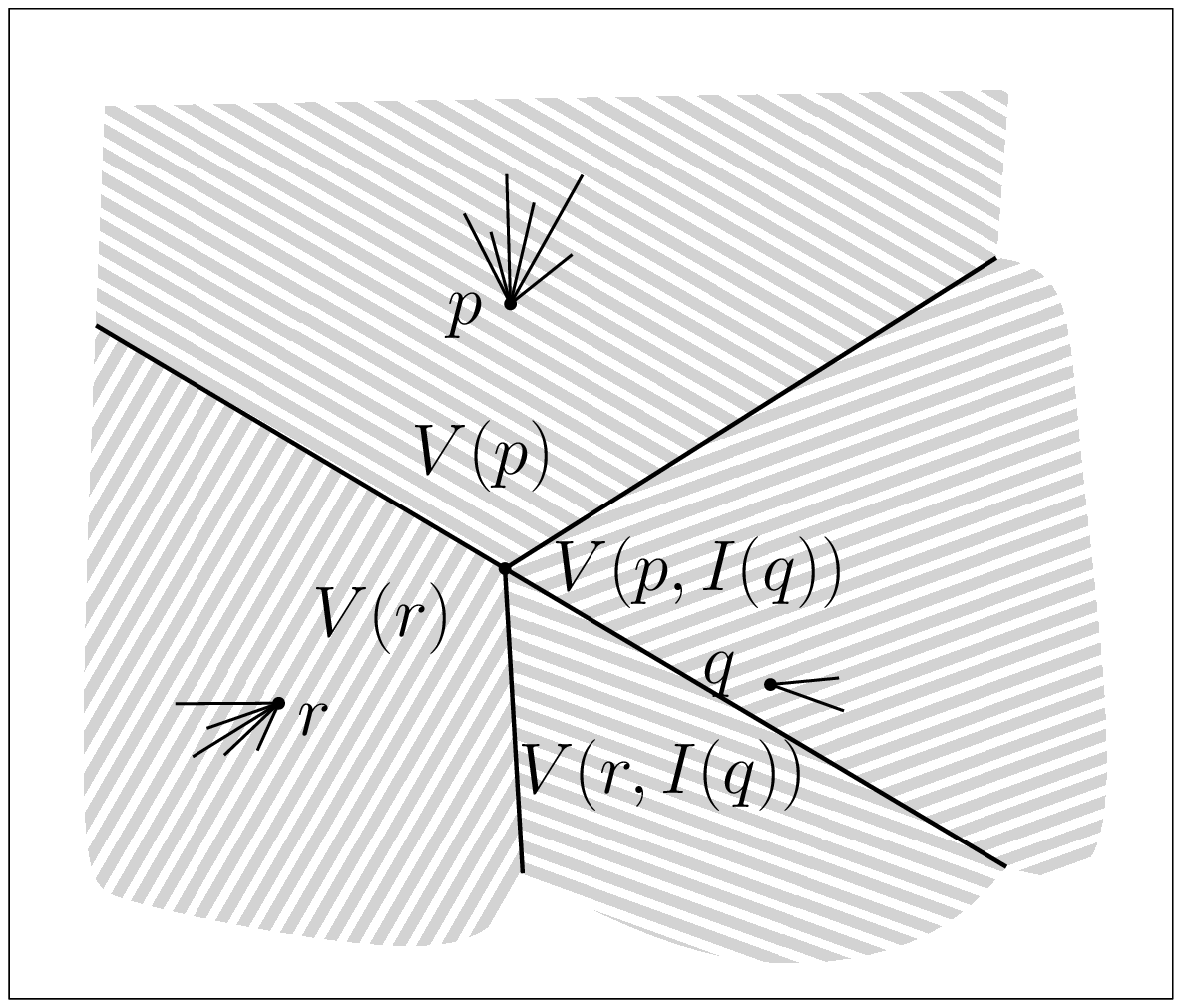}
    \includegraphics[width=1.5 in, clip, trim = 15mm 15mm 29mm 15mm]{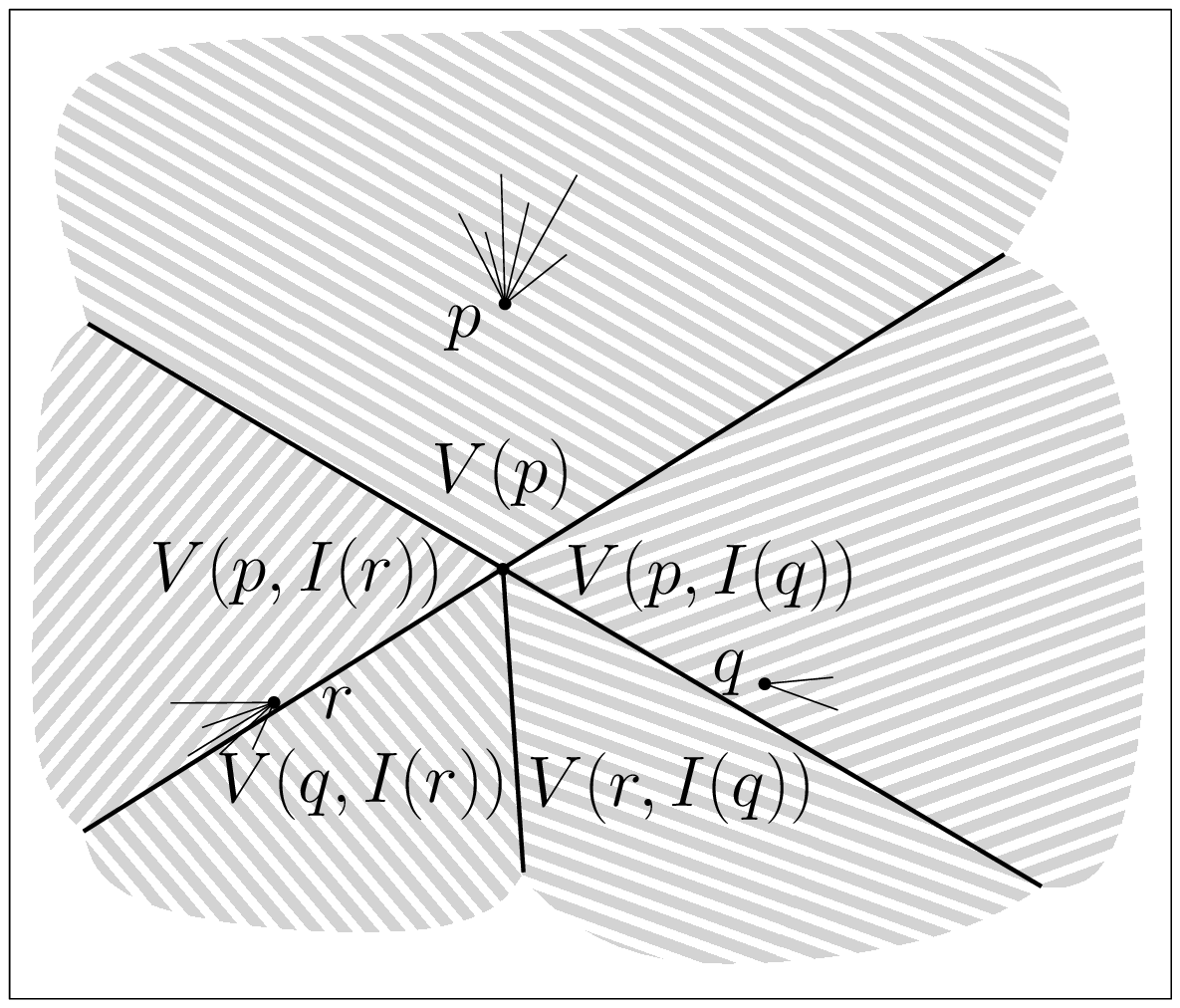}
    \\
    \put(50,0){(a)}
    \put(160,0){(b)}  
    \put(275,0){(c)}  
    \caption{(a) A Voronoi vertex in $\mathcal{V}_1(S)$ incident to three Type-2 regions; 
    (b) In $\mathcal{V}_{j+1}(S)$,  $j=|I(q)|$, region $V(q)$ is split by the representatives of the neighboring Type-2 regions;
    (c) In $\mathcal{V}_{k+1}(S)$,  $k=|I(r)|$, region $V(r)$ is split by the representatives of the neighboring Type-2 regions.
} \label{fig:pslg_vertices2}
  \end{figure}

  \subsection{Structural Complexity and Perturbation}\label{subsection:pslg_complexity}
  Let $S(\varepsilon)$ be the set of disjoint line segments obtained from $S$ by a small perturbation $\varepsilon>0$ of their common endpoints, as shown in Fig.~\ref{fig:pert}(a). 
  In particular, for every endpoint $p$ with $|I(p)|>1$, and for every line segment $s\in I(p)$, move the endpoint  of $s$  incident to $p$ along the line through $s$ and $p$ by a small amount $\delta_s<\varepsilon$ within $n(p,\varepsilon)$, $n(p,\varepsilon)=\{x~|~d(x,p)<\varepsilon\}$.
  By using variable amounts for $\delta_s$, for each segment $s$, and given the weak general-position assumption, the general-position assumption can be easily maintained.
Consider $\mathcal{V}_k(S(\varepsilon))$, see Figure~\ref{fig:pert}. It contains many artificial faces, however, its structural complexity  is $O(k(n-k))$.

In the remaining of this section, we show that the number of faces in $\mathcal{V}_k(S)$  cannot exceed the one of $\mathcal{V}_k(S(\varepsilon))$ for certain $\varepsilon$, and thus, the complexity of $\mathcal{V}_k(S)$ is also $O(k(n-k))$.
  To this aim, we use the refined version of $\mathcal{V}_k(S)$ and $\mathcal{V}_k(S(\varepsilon))$, where all regions are subdivided into the finest subfaces as obtained by superimposing the corresponding order-$(k{-}1)$ diagrams.
  Additionally, the faces of $\mathcal{V}_k(S(\varepsilon))$ are further subdivided by their elementary sites, such that for every point $x$ in a fine face, $D_k(x)=D_k^p(x)$ for exactly one elementary site $p$.

  \begin{lemma}\label{lemma:pslg_eps}
    There is an injection from the (fine) faces of $\mathcal{V}_k(S)$ to the (fine) faces of $\mathcal{V}_k(S(\varepsilon))$, for some  $\varepsilon>0$.
  \end{lemma}
  \begin{proof}
    For every (fine) face $F_j$ of $\mathcal{V}_k(S)$, consider an arbitrary point $a_j$ in its interior.
    Point $a_j$ corresponds to a proper order-$k$ disk $D_k(a_j)$.
    For every elementary site $x$ that does not touch $D_k(a_j)$, let $d(a_j,x)$ be the distance from $x$ to the boundary of disk $D_k(a_j)$.
    If $x$ is an open segment, then $d(a_j,x)$ is defined as the difference by which we have to shrink or expand the disk to make it touch $x$.
    Let $\varepsilon_1=\min_{x,a_j}d(a_j,x)/2$.
    Let $D_k(a_j,\varepsilon)$ denote the order-$k$ disk in $\mathcal{V}_k(S(\varepsilon))$ and let $S_k(a_j,\varepsilon)$ be the set of line segments intersected by it. 
    For $\varepsilon=\varepsilon_1$, and for every point $a_j$, we have $S_k(a_j,\varepsilon)\subseteq S_k(a_j)$.
    We map every point $a_j$ to the face in $\mathcal{V}_k(S(\varepsilon))$ to which it belongs.

    To avoid having more than one point mapped to the same face in $\mathcal{V}_k(S(\varepsilon))$, we do the following.
    Consider point $a_j$ and its corresponding face $F_j$ in $\mathcal{V}_k(S)$.
    Let $\gamma(a_j)$ be a curve passing around $F_j$ without touching $F_j$, but intersecting all the faces adjacent to $F_j$.
    Consider an arbitrary point $y$ on that curve.
    Let $p$ be the elementary site that defines $F_j$, and let $d(y,a_j)$ be the minimum distance from $p$ to the boundary of disk $D_k(y)$. 
    In case $p$ touches $D_k(y)$ (i.e. the face of $y$ is a Type-2 face of representative $p$) let $d(y,a_j)$ be the minimum distance from any other elementary site to the boundary of $D_k(y)$.
    Thus, $d(y,a_j)>0$ for any point $y\in \gamma(a_j)$.
    Let us set $\varepsilon_2=\min_{a_j, y\in \gamma(a_j)}d(y,a_j)/2$ and $\varepsilon=\min\{\varepsilon_1,\varepsilon_2\}$.
    The choice of $\varepsilon_2$ guarantees that the faces of $\mathcal{V}_k(S(\varepsilon))$ intersected by $\gamma(a_j)$ differ from the face in $\mathcal{V}_k(S(\varepsilon))$ assigned to $a_j$.
    This implies that no two points are mapped to the same face in $\mathcal{V}_k(S(\varepsilon))$.
  \end{proof}

  \begin{figure}[h]
    \centering
    \includegraphics[width=1.5 in]{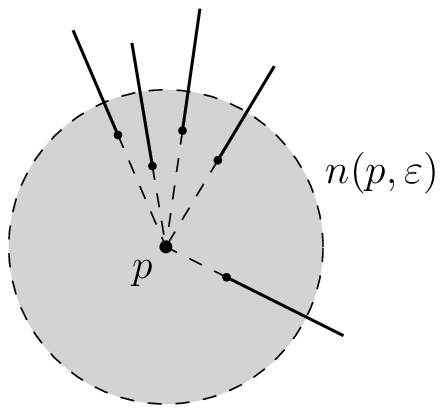}
    \includegraphics[width=3 in, clip, trim = 0mm 0mm 0mm 0mm]{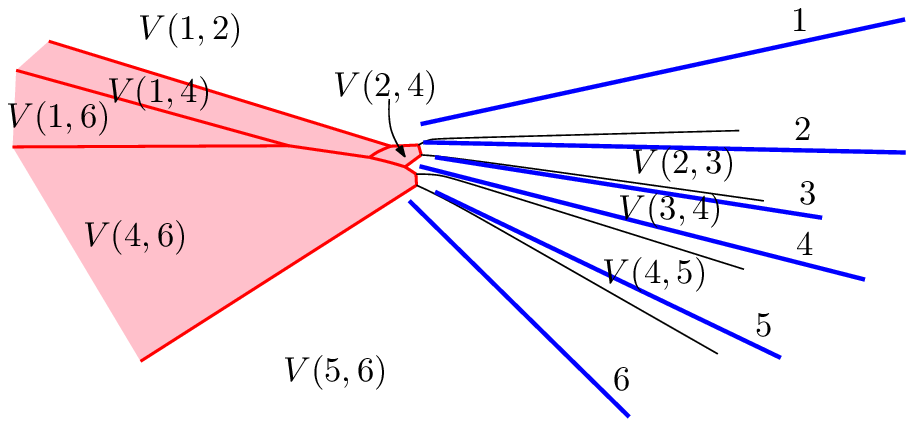}
    \caption{(a) Untangling abutting line segments at endpoint $p$; (b) Order-2 Voronoi diagram of untangled line segments. 
    Artificial edges and faces are shown bold. 
$V(e_1,\cdots,e_m)$ stands for $V_k(H,S)$, where  $\{e_1,\dots,e_m\}=H$.
  }
    \label{fig:pert}
  \end{figure}

  By Lemma~\ref{lemma:pslg_eps}, we conclude.

  \begin{theorem}\label{theorem:bound_pslg}
    The structural complexity of the order-$k$ Voronoi diagram of $n$ line segments forming a planar straight-line graph is $O(k(n-k))$.
  \end{theorem}

  \section{Intersecting Line Segments}\label{section:intersecting_segments}

  In this section we extend our complexity results of Section~\ref{section:structural_complexity} to intersecting line segments with a total of $I$ intersection points, $I=O(n^2)$.
  We show that  segment-intersections influence the Voronoi diagram for small $k$ and  the influence grows weaker as $k$ increases.
  For  $k\geq n/2$, intersections no longer affect the asymptotic complexity of the order-$k$ Voronoi diagram.

In the following, we extend Lemma~\ref{lemma:total_unb}, Theorem~\ref{lemma:lee}, and Theorem~\ref{theorem:bound}  to intersecting line segments as Lemma~\ref{lemma:total_unbI}, Theorem~\ref{theorem:leeI}, and Theorem~\ref{theorem:int_bound}, respectively.
  To simplify the analysis, we assume that no two segments share a common endpoint and that no more than two segments intersect at the same point.
  Recall that the numbers of faces, edges, vertices, and unbounded edges of $\mathcal{V}_k(S)$ are denoted as $F_k$, $E_k$, $V_k$, and $U_k$, respectively.

  \begin{lemma}\label{lemma:total_unbI}
    The total number of unbounded edges in the order-$k$ Voronoi diagram for all orders is $\sum_{i=1}^{n-1}U_i=n(n-1)+2I$
  \end{lemma}
  \begin{proof}
    Consider a pair of line segments.
    If the pair does not intersect, then it defines exactly two open halfplanes, such that  each halfplane induces exactly one unbounded Voronoi edge in $\mathcal{V}_k(S)$ for some order $k$ (see Lemma~\ref{lemma:total_unb}). 
    If the pair intersects, then it induces exactly four  such unbounded Voronoi edges. 
    Thus, each pair of intersecting segments induces exactly two additional unbounded Voronoi edges,  in addition to those counted  in Lemma~\ref{lemma:total_unb}.
    Therefore, the total number of unbounded faces in all orders is $\sum_{i=1}^{n-1}U_i=2{n\choose 2}+2I=n(n-1)+2I$. 
  \end{proof}

  \begin{theorem}\label{theorem:leeI}
    The number of faces in the order-$k$ Voronoi diagram of a set $S$ of $n$ line segments with $I$ intersections is:  
    \begin{eqnarray}
      F_k=2kn-k^2-n+1-\sum_{i=1}^{k-1}U_i+2I \label{eq:app_leeI}\\
      \mbox{or equivalently~} F_k=1-(n-k)^2+\sum_{i=k}^{n-1}U_i \label{eq:app_dual_leeI}
    \end{eqnarray}
  \end{theorem}
  \begin{proof}
    Consider the partitioning of segments into pieces as obtained by their intersection points.
    Every component of a segment induces exactly one face in $\mathcal{V}_1(S)$, thus, $\mathcal{V}_1(S)$ has two types of vertices: 
    (1) $I$ {\em intersection}\/ points, which are incident to exactly four Voronoi edges each; and 
    (2)  $V_1-I$ {\em regular}\/ Voronoi vertices, which are incident to three Voronoi edges each (under the general position assumption).
    Regular Voronoi vertices are the \emph{new vertices}\/ of $\mathcal{V}_1(S)$, thus, $V_1'=V_1-I$.

    Consider the dual graph of $\mathcal{V}_1(S)$, augmented with a vertex at infinity to connect the dual of unbounded faces.
    Using standard arguments,  $2E_1=4I+3(V_1-I)+U_1$ (see also the proof of Theorem~\ref{theorem:bound}). 
    Note that the dual graph consists of faces of four edges each that correspond to {\em intersections}, and faces of three edges each that correspond to {\em regular}\/ Voronoi vertices of $\mathcal{V}_1(S)$. 
    Euler's formula  and the latter equation imply  $E_1=3(F_1-1)-U_1-I$.  
    Thus, $E_1=3n-3-U_1+5I$. 
    By Euler's formula, $V_1=1+E_1-F_1=1+E_1-n-2I$, thus, $V_1=2n-2-U_1+3I$. 

    Consider now $\mathcal{V}_2(S)$, which has two types of faces: 
    faces that contain exactly one edge of $\mathcal{V}_1(S)$ and faces that contain an {\em intersection}\/ point of $\mathcal{V}_1(S)$. 
    As a result, the total number of faces in $\mathcal{V}_2(S)$ is $F_2=(E_1-4I)+I=E_1-3I$. 
    Therefore, $F_2=3(F_1-1)-U_1-4I=3(n-1)-U_1+2I$. 
    Since all Voronoi vertices of $\mathcal{V}_2(S)$ have degree three, Lemma~\ref{lemma:euler} implies that $E_2=3F_2-3-U_2$.
    Plugging in  the formula for $F_2$, we obtain $E_2=9n-12-3U_1-U_2+6I$. 

    For an order $i$-diagram, $i\geq 3$, every vertex of the diagram and every vertex of the farthest subdivision is incident to exactly 3 edges, and thus, Claim~1 in the proof of Theorem~\ref{lemma:lee} and its proof remain identical. 
    Thus, the recursive formula of Eq.~(\ref{eq:iterative}) remains valid for any $k\geq 1$.

    Using Claim~1 of Theorem~\ref{lemma:lee}, $F_3=E_2-2V_1'=E_2-2(V_1-I)$. 
    Plugging in the formulas obtained for $E_2$ and $V_1$, we obtain $F_3=5n-8-U_1-U_2+2I$. 

    Since the recursive formula in Eq.~(\ref{eq:iterative}) remains valid for any $k\geq 1$, we can use induction, with bases cases the above formulas for $F_2$ and $F_3$, and derive Eq.~(\ref{eq:app_leeI}).
    Note that the main difference with the derivation of Theorem~\ref{lemma:lee} are the base cases $F_1,F_2$, and $F_3$, where $F_3$ is no longer obtained by Eq.~(\ref{eq:iterative}). 
    Then Eq.~(\ref{eq:app_dual_leeI}) can be derived from Eq.~(\ref{eq:app_leeI}) using Lemma~\ref{lemma:total_unbI}.
  \end{proof}

  \begin{theorem}\label{theorem:int_bound} 
    The combinatorial complexity of the order-$k$ Voronoi diagram of $n$  properly intersecting line segments with $I$ intersections is  $O(k(n-k)+I)$ for $1\leq k< n/2$, and $O(k(n-k))$ for $n/2\leq k\leq n-1$.
  \end{theorem}
  \begin{proof}
    For $1\leq k<n/2$, Eq.~(\ref{eq:app_leeI}) of Theorem~\ref{theorem:leeI} directly implies $F_k=O(k(n-k)+I)$.
    The proof of Lemma~\ref{lemma:unb} remains valid for any set of arbitrary line segments, including intersecting ones.
    Thus, for $n/2\leq k\leq n-1$, Eq.~(\ref{eq:app_dual_leeI}) of Theorem~\ref{theorem:leeI} and Lemma~\ref{lemma:unb} imply $F_k=O(k(n-k))$. 
  \end{proof} 

  \section{The Iterative Construction}\label{section:algorithms}
  To compute the diagram, we can use the standard iterative approach to construct higher-order Voronoi diagrams (see e.g.,~\cite{Lee82}) and enhance it with the ability to  deal with disconnected  regions that are present in the case of line segments.
  Although not very efficient for arbitrary $k$, the iterative construction is basic and it is also valuable to applications, where lower order diagrams  are required to be computed in any case, such as in~\cite{Papadopoulou11}.

  Given $\mathcal{V}_i(S)$, the iterative construction considers every face $F$ of every region  $V_i(H,S)$ and computes  $\mathcal{V}_1(S\setminus H)$ within the interior of $F$; this gives the order-$(i{+}1)$ subdivision within $F$.
  For a PSLG, only faces of Type-1 need to be considered because faces of Type-2 contain no portions of the order-$(i{+}1)$ diagram. 
  Then, the iterative construction merges any two neighboring order-$(i{+}1)$ faces that belong to the same $(i{+}1)$-subset and removes the corresponding portion of  the boundary of $F$. 
  In case of disjoint sites, all edges of $\mathcal{V}_i(S)$ are removed to obtain, $\mathcal{V}_{i+1}(S)$.
  In case of a PSLG, edges incident to Type-2 Voronoi regions remain for several orders.

  Given a face $F$ of region $V_i(H,S)$, let $S_F$ denote the collection of segments in $S\setminus H$ that define a Voronoi edge along the boundary of $F$, $\partial F$. 
  In case of a PSLG, let $F$ be a face of Type-1.
  Let $\mathcal{V}_1(F)$ denote the portion of $\mathcal{V}_1(S_F)$ in the interior of $F$.
  By the definition of an order-$(i{+}1)$ region, $\mathcal{V}_1(F)$ corresponds exactly to $\mathcal{V}_1(S\setminus H)$ within $F$.
  The main operation of the iterative construction is to compute $\mathcal{V}_1(F)$.
  Figure~\ref{fig:iterative_disconnected} illustrates an unbounded  face $F$ shaded and $\mathcal{V}_1(F)$ in dashed lines. 
  Because $F$ is  unbounded, $\mathcal{V}_1(F)$  is augmented with an artificial point at infinity, which is assumed to be incident to all unbounded Voronoi edges.

  Because order-$k$ Voronoi regions may be disconnected, a segment $s\in S_F$ may appear multiple times along $\partial F$.
  In fact, it may appear $\Theta(|S_F|)$ times as illustrated  in Figure~\ref{fig:iterative_disconnected}.
  Nevertheless, $\mathcal{V}_1(F)$ remains a tree structure as illustrated in the following lemma.
  Given a point $x$ in a face $P$ of the region of $s\in S_F$  in $\mathcal{V}_1(F)$, 
  let $r(s,x)$ denote the ray that realizes the distance from segment $s$ to $x$, emanating away from $s$ (see Figure~\ref{fig:iterative_disconnected}). 

  \begin{lemma} 
    \label{lemma:boundary_visibility_property}
    The graph structure of $\mathcal{V}_1(F)$ is a tree\footnote{In case of an unbounded face $F$, we assume an artificial  vertex at infinity incident to all unbounded edges}. 
    Furthermore, every face $P$ of $\mathcal{V}_1(F)$ has the following visibility property: for every point $x$ in $P$, open segment $xa_x$,  where $a_x$ is the point along the boundary of $F$ first intersected by $r(s,x)$, lies entirely in $P$.
  \end{lemma}
  \begin{proof}
    Let $D_{i+1}(x)$ be the order-$(i{+}1)$ disk centered at point $x$ in $P$.
    $D_{i+1}(x)$ touches segment $s$ and intersects all segments in $H$. 
    Let $y$ be an arbitrary point along segment $xa_x$.
    Since $y\in F$, disk $D_{i+1}(y)$ must intersect all line segments in $H$.
    Furthermore, since $y$ is closer to $s$ than $x$ and $D_{i+1}(x)$ touches $s$, $D_{i+1}(y)$ must also touch $s$.
    Thus, $y\in P$.
    Since $y$ is taken arbitrarily, the segment $xa_x$ must lie entirely in $P$.

    Since every face of  $\mathcal{V}_1(F)$ must touch $\partial F$, the graph structure $T$ of $\mathcal{V}_1(F)$  must be a tree or a forest.
    To prove that $T$ is a tree it is enough to show that every occurrence of a segment $s\in S_F$ along $\partial F$ corresponds to a distinct face of  $\mathcal{V}_1(F)$.
    To this aim, consider a point $y$ on $\partial F$ between two consecutive occurrences of segment $s$ on $\partial F$.
    Ray $r(s,y)$ cannot intersect any face $P$ of $s$ because for any point $x$ along the portion of $r(s,y)$ in $P$ segment $xa_x$ is not entirely contained in $P$. 
    Thus, if $x$ is in a face of $s$ the above visibility property would not hold for $x$.
    Thus, the two distinct occurrences of $s$ along $\partial F$ must correspond to distinct faces of $s$ at opposite sides of $r(s,y)$.
    Therefore, $T$ must be a tree.
  \end{proof}

  \begin{figure}
    \begin{center}
      \includegraphics[width=4 in, clip, trim = 7mm 21mm 15mm 0mm]{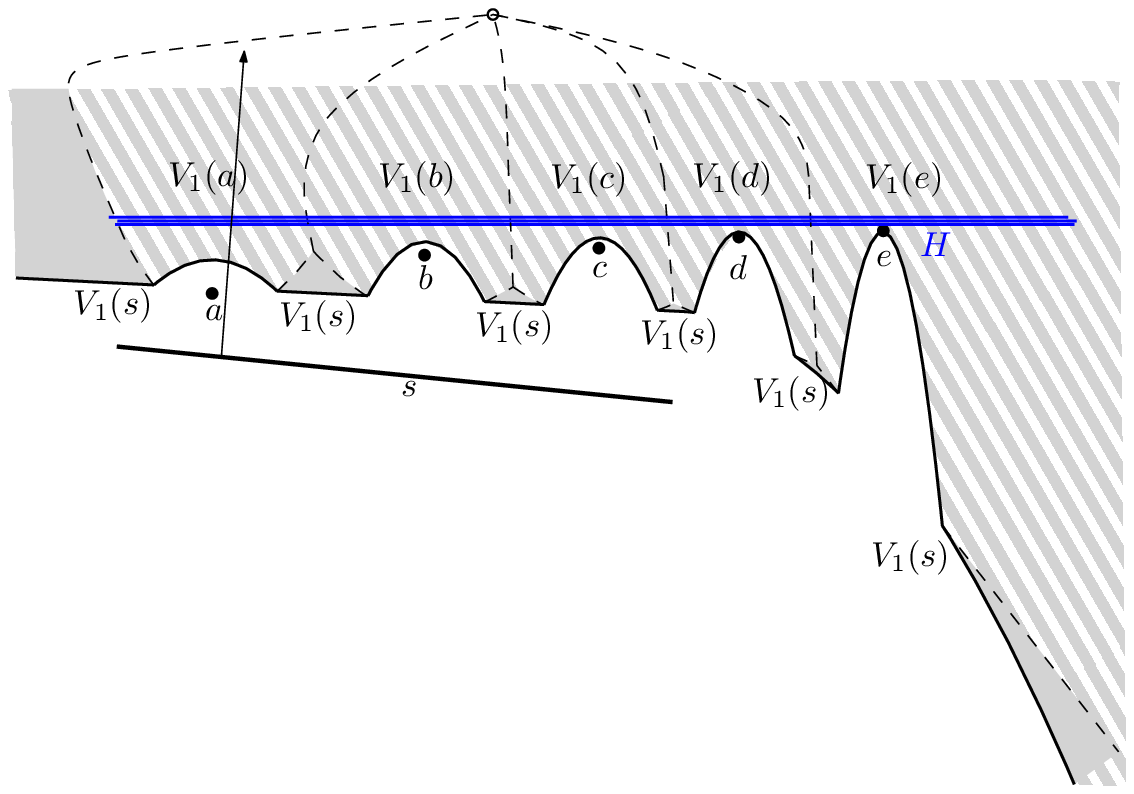}
    \end{center}
    \caption{A face of an order-$i$ Voronoi region induced by set $H$ of line segments, for $i=3$.
    The segment $s$ contributes linear number of subfaces. } 
    \label{fig:iterative_disconnected}
  \end{figure}

  $\mathcal{V}_1(F)$ can be easily computed in $O(|S_F|\log |S_F|+ |\partial F|)$ time  by first computing $\mathcal{V}_1(S_F)$  and then truncating it within the interior of $F$.
  This results in the standard $O(k^2n\log{n})$-time iterative construction.
  The space complexity is proportional to the size of the largest order-$i$ Voronoi diagram, for $i=1,\ldots,k$.
  We conjecture that $\mathcal{V}_1(F)$ can also be computed in linear  time, linear on the complexity of $\partial F$.
  However, this goes beyond the scope of this paper and we leave it as a topic of our 
future  research. 
  For points, 
 $\mathcal{V}_1(S_F)$ can be computed in linear time as claimed in \cite{Aggarwal89}.
  To adapt \cite{Aggarwal89} for the case of line segments, the issue of the multiplicity of sites along $\partial F$ must be resolved efficiently.
  This issue for the farthest line-segment Voronoi diagram, in a randomized linear construction after the cyclic sequence of faces at infinity was computed, was recently resolved in~\cite[journal version]{Papadopoulou12}. 
Simiar techniques are likely to be applicable to this problem as well.

  \section{Extending to the $L_p$ Metric}\label{section:lp}

  The results of Sections~\ref{section:structural_complexity}, \ref{section:abutting_segments}, and~\ref{section:intersecting_segments}, extend naturally to the general $L_p$, $1\leq p\leq \infty$,  metric.

  An $L_p$ disk of infinite radius, $1<p<\infty$, is an ordinary halfplane~\cite{Lee80}, thus, Lemma~\ref{lemma:supp}, Def.~\ref{def:supporting-halfplane} and Corollary~\ref{cor:unb_edge} remain identical in all these metrics.
  As a result, Lemmas~\ref{lemma:total_unb} and~\ref{lemma:total_unbI} also remain identical.  
  On the other hand, Lemmas~\ref{lemma:face_tree}-\ref{lemma:tree} in Section~\ref{section:structural_complexity} never make explicit use of the Euclidean metric, and  they can be easily extended to  $L_p$ for  all $p$, $ 1\leq p \leq \infty$.
  Thus, the formulas of Theorem~\ref{lemma:lee} and the $O(k(n-k))$ complexity bound of Theorem~\ref{theorem:bound} remain the  same in $L_p$ for  $1< p <\infty$.
  Similarly for Lemma~\ref{lemma:total_unbI}, and Theorems~\ref{theorem:leeI},~\ref{theorem:int_bound}, in case of intersecting line segments.
  Thus, all structural properties of the order-$k$ Voronoi diagram in the Euclidean metric remain the same in $L_p$, for  $ 1< p <\infty$.
  
In the remaining of this section, we extend our results to the $L_\infty$ metric (equiv. $L_1$).
  In the $L_\infty$ metric, the equivalent of a supporting halfplane (see Def.~\ref{def:supporting-halfplane}) is a \emph{supporting quadrant}.
  A \emph{quadrant} is the common intersection of two halfplanes, which are defined by axis parallel perpendicular lines.
  Thus, Corollary~\ref{cor:unb_edge} is adapted as follows: 
  There is un unbounded Voronoi edge separating the $L_\infty$ unbounded regions $V_k(H\cup\{s_1\},S)$ and $V_k(H\cup\{s_2\},S)$ if and only if there is an open quadrant that touches $s_1$ and $s_2$, intersects all line segments in $H$, but no line segment in $S\setminus H$. 
  Such a quadrant is called a {\em supporting quadrant}\/ (see e.g., Fig.~\ref{fig:quadrants}).

  In $L_\infty$, a pair of disjoint line segments admits  two supporting quadrants and a pair of intersecting line segments admits four supporting quadrants. 
  Thus,  Lemmas~\ref{lemma:total_unb} and~\ref{lemma:total_unbI} remain valid.
  We now extend Lemma~\ref{lemma:unb} to the $L_\infty$  metric. 

  \begin{figure}
    \begin{center}
      \includegraphics[width=3 in, clip, trim = 0mm 0mm 0mm 0mm]{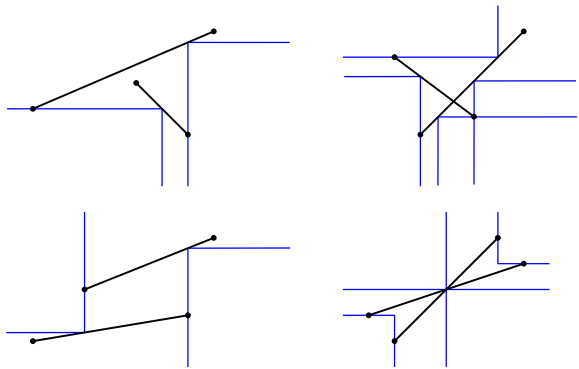}
    \end{center}
    \caption{Examples of supporting quadrants of pairs of line segments in the $L_\infty$ metric.}
    \label{fig:quadrants}
  \end{figure}

  \begin{lemma}\label{lemma:unbLp}
    In $L_\infty$ (resp. $L_1$), for a given set of $n$ line segments,  $\sum_{i=k}^{n-1}U_i=O(n(n-k))$. 
    If segments are disjoint then $\sum_{i=k}^{n-1}U_i=O\left( (n-k)^2 \right)$.
  \end{lemma}
  \begin{proof}
    The duality transformation in the proof of Lemma~\ref{lemma:unb} is not extendible to the $L_\infty$ metric.
    Instead, we use the abstract framework presented in~\cite{Clarkson87,Clarkson89,Sharir03}.

    Let a supporting quadrant be called  a {\em configuration}.
    A configuration is defined by two line segments $s_1$ and $s_2$ if there is a quadrant whose boundary touches $s_1, s_2$ and its interior does not intersect $s_1, s_2$.
    A configuration is said to be {\em in conflict}\/ with line segment $s'$ if its supporting quadrant does not intersect $s'$.
    The \emph{weight} of a configuration is the number of its conflicts.
    The maximum number of configurations of weight $i$ in a set of $n$ line segments is denoted as $N_i(n)$, and the maximum number of configurations of weight at most $i$ is denoted as $N_{\leq i}(n)$. 
    The configurations with weight $i$ correspond to unbounded Voronoi edges in the order-$(n{-}i{-}1)$ Voronoi diagram, thus $U_{n-i-1}\leq N_i(n)$.
    The configurations with weight 0 correspond to unbounded edges in the farthest Voronoi diagram.
    The Clarkson-Shor abstract framework implies $N_{\leq i}(n)=O\left( i^2N_0(n/i) \right)$.
    Substituting $i=n-k-1$, we derive 
    \begin{equation}
      \sum_{i=k}^{n-1}U_i\leq N_{\leq n-k-1}(n)=O\left( (n-k-1)^2N_0\left( \frac{n}{n-k-1} \right) \right)\label{eq:f05}
    \end{equation}

    In $L_\infty$, $N_0(n)$ is $O(n)$ for arbitrary line segments, and $O(1)$ for non-crossing line segments~
    \cite{Papadopoulou12,Dey12}.
    Substituting these values in Eq.~(\ref{eq:f05}), we derive $\sum_{i=k}^{n-1}U_i=O(n(n-k))$ for arbitrary line segments, and $\sum_{i=k}^{n-1}U_i=O\left( (n-k)^2 \right)$ for non-crossing line segments.
  \end{proof}

  Using  Lemma~\ref{lemma:unbLp} in place of Lemma~\ref{lemma:unb}, we can extend the proofs of Theorems~\ref{theorem:bound} and~\ref{theorem:int_bound} to the $L_\infty$ metric in a straightforward way.   
  For  non-crossing line segments, Lemma~\ref{lemma:unbLp} directly implies a tighter bound.
  The same tighter bound was shown for points in~\cite{Liu11} by a different derivation based on a \emph{Hanan grid}, which is not applicable to line segments.
  We summarize in the following theorem.

  \begin{theorem}\label{theorem:Linfinity}
    The structural complexity of order-$k$ Voronoi diagram of $n$ arbitrary  line segments, with  $I$ intersections, in the $L_p$ metric, $1\leq p\leq \infty$, is:
    \begin{alignat*}{1}
	O\left(k(n-k)+I\right) \text{,}&\quad\text{for $1\leq k< n/2$;}\\
	O\left(k(n-k)\right) \text{,}&\quad\text{for $n/2\leq k\leq n-1$;}\\
	O\left((n-k)^2\right) \text{,}&\quad\text{for $n/2\leq k\leq n-1$, non-crossing segments 
	and $p=1,\infty$.}
     \end{alignat*}
  \end{theorem}

  \section{Concluding Remarks}\label{section:conclusion}
  The higher-order Voronoi diagram of line segments had been surprisingly ignored in the computational geometry literature. 
  In this paper, we analyzed its structural properties and showed that despite the presence of disconnected Voronoi regions, the combinatorial complexity remains $O(k(n-k))$ (assuming non-crossing line segments). 
  For intersecting line segments, the influence  of intersections grows weaker as $k$ increases and the complexity of the diagram remains $O(k(n-k))$ for $k\geq n/2$.
  The case of a planar straight line graph required to augment the definition of an order-$k$ diagram to include non-disjoint sites. 
  The diagram can be constructed in $O(k^2n\log{n})$ time for non-crossing segments using the standard  iterative construction. 
  We conjecture that it can also be computed in $O(k^2n+n\log{n})$ time, as in the case of points.
  This requires a linear time algorithm to compute the order-$(k{+}1)$ subdivision ($\mathcal{V}_1(F)$) within an order-$k$ face $F$. 
  We showed that $\mathcal{V}_1(F)$  is a tree structure, however,  due to the presence of disconnected regions in the order-$k$ diagram, the region of a segment in $\mathcal{V}_1(F)$ may 
consist of multiple disjoint faces.
  These issues create complications to the linear construction of $\mathcal{V}_1(F)$, which we plan to investigate in future research.
  The iterative approach is efficient  for small values of $k$, and it is valuable to applications where lower orders are required to be computed in any case.
  In future research, we also plan to investigate algorithmic techniques appropriate for larger values of~$k$.

  \section{Acknowledgements}
  The authors are grateful to Gill Barequet for reading through an earlier draft and 
  providing useful comments.

  This work was supported in part by the Swiss National Science Foundation grant 200021-127137  and the ESF EUROCORES program EuroGIGA / VORONOI, SNF 20GG21-134355.

  \end{document}